\documentclass{birkjour}

\usepackage[utf8]{inputenc}
\usepackage{pdfsync}
\usepackage[T1]{fontenc}
\usepackage{color}
\usepackage{graphics}
\usepackage[ngerman,english]{babel}

\usepackage{amsmath}
\usepackage{amsopn}

\usepackage{units}
\usepackage[all,cmtip]{xy}
\include{diagxy}
\providecommand{\R}{\mathbb{R}}

\providecommand{\C}{\mathbb{C}}
\renewcommand{\C}{\mathbb{C}}

\providecommand{\N}{\mathbb{N}}

\providecommand{\eps}{\varepsilon}

\providecommand{\ran}{\mathrm{ran} \, }

\providecommand{\abs}[1]{\left \lvert #1 \right \rvert}
\providecommand{\sabs}[1]{\lvert #1 \vert}
\providecommand{\babs}[1]{\bigl \lvert #1 \bigr \rvert}

\providecommand{\norm}[1]{\left \lVert #1 \right \rVert}
\providecommand{\snorm}[1]{\lVert #1 \rVert}
\providecommand{\bnorm}[1]{\bigl \lVert #1 \bigr \rVert}
\providecommand{\Bnorm}[1]{\Bigl \lVert #1 \Bigr \rVert}

\providecommand{\dd}{\mathrm{d}}
\providecommand{\id}{\mathrm{id}}

\providecommand{\order}{\mathcal{O}}
\providecommand{\Fourier}{\mathcal{F}}

\providecommand{\ie}{i.~e.~}
\providecommand{\eg}{e.~g.~}

\newtheorem{thm}{Theorem}[section]
\newtheorem{defn}[thm]{Definition}

\newtheorem{lem}[thm]{Lemma}
\newtheorem{cor}[thm]{Corrolary}
\newtheorem{prop}[thm]{Proposition}
\newtheorem{assumption}[thm]{Assumption}

\theoremstyle{remark}
\newtheorem{remark}[thm]{Remark}

\usepackage[charter]{mathdesign} % math font has ugly \mathcals
\SetMathAlphabet{\mathcal}{normal}{OMS}{cmsy}{m}{n} % fixes ugly \mathcals
\SetMathAlphabet{\mathcal}{bold}{OMS}{cmsy}{m}{n} % fixes ugly \mathcals
\usepackage{helvet}

\usepackage{enumerate}

\usepackage[style=alphabetic,citestyle=alphabetic,firstinits=true,backend=biber,hyperref=auto,sorting=nyt,maxnames=4,arxiv=pdf,url=false,doi=false,isbn=false,dateabbrev=true,datezeros=true,bibencoding=utf8,babel=hyphen]{biblatex}
\DeclareFieldFormat
	[article,inbook,incollection,inproceedings,patent,thesis,unpublished]
	{title}{\emph{#1\isdot}}
\renewbibmacro*{in:}{}%   löscht alle 'In:'
\DeclareFieldFormat
	[article]
	{journaltitle}{#1\isdot}
\DeclareFieldFormat
	[article]
	{volume}{\textbf{#1}}
\bibliography{/Users/max/Library/texmf/tex/latex/max/max_bibliography.bib}

\numberwithin{equation}{section}

\providecommand{\Hoer}[1]{S^{#1}}
\providecommand{\Hoerr}[1]{S^{#1}_{\rho}}

\providecommand{\Hoermr}[2]{S^{#1}_{#2}}

\providecommand{\SemiHoer}[1]{\Hoer{#1}_{\mathrm{sr}}}

\providecommand{\Hd}{H_{\mathrm{D}}}

\providecommand{\Weyl}{\sharp}
\providecommand{\WeylProd}{\sharp}

\providecommand{\Op}{\mathrm{Op}}

\providecommand{\heff}{h_{\mathrm{eff}}}

\providecommand{\orderc}[1]{\mathcal{O}(\nicefrac{1}{c^{#1}})}
\providecommand{\heff}{{h_{\mathrm{eff}}}}
\providecommand{\spec}{\sigma}

\providecommand{\magW}{\WeylProd^B}

\providecommand{\Qe}{\mathsf{Q}}
\providecommand{\Fourier}{\mathfrak{F}}
\providecommand{\Cont}{\mathcal{C}}

\providecommand{\Weyl}{\mathrm{W}^M_{\varepsilon}}

\providecommand{\Piref}{\Pi_{\mathrm{ref}}}

\providecommand{\Hfast}{\mathcal{H}_{\mathrm{fast}}}
\providecommand{\Hslow}{\mathcal{H}_{\mathrm{slow}}}

\providecommand{\specrel}{\sigma_{\mathrm{rel}}}
\providecommand{\piref}{\pi_{\mathrm{ref}}}

\providecommand{\Hslow}{\mathcal{H}_{\mathrm{slow}}}
\providecommand{\Hfast}{\mathcal{H}_{\mathrm{fast}}}

\providecommand{\Qe}{{Q_{\eps}}}

\providecommand{\Hoer}[1]{S^{#1}}
\providecommand{\Hoerr}[1]{S^{#1}_{\rho}}

\providecommand{\Hoermr}[2]{S^{#1}_{#2}}

\providecommand{\SemiHoer}[1]{\Hoer{#1}_{\mathrm{sr}}}

\providecommand{\order}{\mathcal{O}}
\providecommand{\Fourier}{\mathfrak{F}}

\providecommand{\Psr}{\mathsf{P}_{\mathrm{sr}}}
\providecommand{\Qsr}{\mathsf{Q}}
\providecommand{\Opsr}{\Op_{\mathrm{sr}}}
\providecommand{\magWsr}{\magW_{\mathrm{sr}}}

\providecommand{\Pnr}{\mathsf{P}_{\mathrm{nr}}}
\providecommand{\Qnr}{\mathsf{Q}}
\providecommand{\Opnr}{\Op_{\mathrm{nr}}}
\providecommand{\magWnr}{\magW_{\mathrm{nr}}}

\providecommand{\Hsr}{H_{\mathrm{sr}}}
\providecommand{\Hnr}{H_{\mathrm{nr}}}

\renewcommand{\C}{\mathbb{C}}

\begin{document}

\date{\today}

\title[Semi- and Non-relativistic Limit of the Dirac Dynamics]{Semi- and Non-relativistic Limit of the Dirac Dynamics with External Fields}

\author{Martin L. R. Fürst}
\address{%
Excellence Cluster Universe \\
Technische Universität München \\
Boltzmannstr. 2\\ 
85748 Garching \\
Germany 
\medskip
\\
\emph{and}
\medskip
\\
Technische Universität München\\ 
Department für Mathematik \\
Boltzmannstraße 3 \\
85747 Garching \\
Germany
}
\email{mfuerst@ma.tum.de}

\author{Max Lein}
\address{%
Eberhard Karls Universität Tübingen\\
Mathematisches Institut \\
Auf der Morgenstelle 10\\
72076 Tübingen\\
Germany
}
\email{lein@ma.tum.de}

\maketitle
\begin{abstract}
	We show how to approximate Dirac dynamics for electronic initial states by semi- and non-relativistic dynamics. To leading order, these are generated by the semi- and non-relativistic Pauli hamiltonian where the kinetic energy is related to $\sqrt{m^2 + \xi^2}$ and $\frac{1}{2m} \xi^2$, respectively. Higher-order corrections can in principle be computed to any order in the small parameter $\nicefrac{v}{c}$ which is the ratio of typical speeds to the speed of light. Our results imply the dynamics for electronic and positronic states decouple to any order in $\nicefrac{v}{c} \ll 1$. 
	
	To decide whether to get semi- or non-relativistic effective dynamics, one needs to choose a scaling for the kinetic momentum operator. Then the effective dynamics are derived using space-adiabatic perturbation theory by Panati et.~al with the novel input of a magnetic pseudodifferential calculus adapted to either the semi- or non-relativistic scaling. 
\end{abstract}

\tableofcontents

%!TEX root = /Users/max/Dropbox/research/non- and semi-relativistic limit of the Dirac equation (FL 2008)/non_semi_rel_Dirac.tex
\section{Introduction} % (fold)
\label{intro}
The quantum dynamics of a relativistic spin-$\nicefrac{1}{2}$ particle subjected to an electric field $\mathbf{E}$ and a magnetic field $\mathbf{B}$ is governed by the Schrödinger-Dirac equation, 
\begin{align}
	i \partial_t \Psi(t) = \widehat{H}_{\mathrm{D}} \Psi(t) 
	, 
	&&
	\Psi(0) = \Psi_0 \in L^2(\R^3,\C^4) 
	, 
	\label{intro:eqn:Dirac_equation}
\end{align}
where 
\begin{align}
	\widehat{H}_{\mathrm{D}} = c^2 \, m \beta - i c \eps \nabla_x \cdot \alpha - \mathrm{e} A(\hat{x}) \cdot \alpha + \mathrm{e} V(\hat{x}) \, \id_{\C^4} 
	\label{intro:eqn:H_D_original}
\end{align}
is the Dirac hamiltonian including potentials $A = (A_1 , A_2 , A_3)$ and $V$ for the magnetic field $\mathbf{B} = \nabla_x \wedge A$ and the electric field $\mathbf{E} = - \nabla_x V$. Besides the physical constants, the speed of light $c$, the particle mass $m$, the semiclassical parameter $\eps$ and the charge $\mathrm{e}$, four $4 \times 4$ matrices enter, namely 
\begin{align*}
	\beta = \left (
	\begin{matrix}
		\id_{\C^2} & 0 \\
		0 & - \id_{\C^2} \\
	\end{matrix}
	\right )
\end{align*}
and 
\begin{align*}
	\alpha_j = \left (
	\begin{matrix}
		0 & \sigma_j \\
		\sigma_j & 0 \\
	\end{matrix}
	\right )
\end{align*}
defined in terms of the Pauli matrices $\sigma_j$ for $j = 1 , 2 , 3$. The shorthand $\alpha = (\alpha_1 , \alpha_2 , \alpha_3)$ denotes a vector of matrices so that $\xi \cdot \alpha := \sum_{j = 1}^3 \xi_j \, \alpha_j$ yields again a $4 \times 4$ matrix. 

In the absence of electromagnetic fields, $\widehat{H}_{\mathrm{D}}^0 := \widehat{H}_{\mathrm{D}} \vert_{A = 0 , V = 0}$ can be diagonalized explicitly (cf.~equation~\eqref{scaling:eqn:h_sr_0}), and the spectrum 
\begin{align*}
	\sigma \bigl ( \widehat{H}_{\mathrm{D}}^0 \bigr ) = \sigma_{\mathrm{ac}} \bigl ( \widehat{H}_{\mathrm{D}}^0 \bigr ) 
	= (-\infty,-m] \cup [+m,+\infty)
\end{align*}
is purely absolutely continuous and consists of a positive and a negative part \cite[Chapter~1.4.4]{Thaller:Dirac_equation:1992}. States in $\mathrm{ran} \, 1_{[0,+\infty)} \bigl ( \widehat{H}_{\mathrm{D}}^0 \bigr )$ are called \emph{particle states} while those associated to the negative energies are anti-particle states; we will usually refer to them as \emph{electronic} and \emph{positronic} states. 

Our goal is to approximate the full dynamics $e^{- i t \widehat{H}_{\mathrm{D}}}$ for electronic initial states by effective hamiltonians $H_{\mathrm{eff}} = \sum_{n = 1}^k \tfrac{1}{c^k} \, H_{\mathrm{eff} \, n}$ in the following sense: for some fixed $k \in \N$ 
\begin{align}
	e^{- i t \widehat{H}_{\mathrm{D}}} \, \Pi = U^* \, e^{- i t H_{\mathrm{eff}}} \, U \, \Pi + \order_{\norm{\cdot}}(\nicefrac{1}{c^k}) 
	\label{intro:eqn:approximating_dynamics}
\end{align}
holds. Here, $\order_{\norm{\cdot}}(\nicefrac{1}{c^k})$ implies that the difference between the two time evolutions in $\order(\nicefrac{1}{c^k})$ in operator norm,  $\Pi = \Pi^2 + \order_{\norm{\cdot}}(\nicefrac{1}{c^k})$ is an almost projection, $U$ a unitary operator up to $\order_{\norm{\cdot}}(\nicefrac{1}{c^k})$. 

We identify two choices for effective hamiltonians: a semi-relativistic and a non-relativistic effective hamiltonian where the kinetic energies are related to $\sqrt{m^2 + \xi^2}$ and $\frac{1}{2m} \xi^2$, respectively. These are derived by making a \emph{choice of scaling of the kinetic momentum operator:} either one writes $\widehat{H}_{\mathrm{D}}$ in terms of the \emph{semi-relativistic} kinetic momentum operator 
\begin{align*}
	\Psr^A = - i \tfrac{\eps}{c} \nabla_x - \tfrac{\mathrm{e}}{c^2} A(\hat{x}) 
\end{align*}
or in terms of \emph{non-relativistic} kinetic momentum 
\begin{align*}
	\Pnr^A = - i \eps \nabla_x - \tfrac{\mathrm{e}}{c} A(\hat{x}) 
	= c \Psr^A
	. 
\end{align*}
Taking $\Psr^A$ as kinetic energy operator yields semi-relativistic effective dynamics (Theorem~\ref{semirel:thm:semi-relativistic_limit}) while starting with $\Pnr^A = \tfrac{1}{c} \Psr^A$, one obtains non-relativistic effective dynamics (Theorem~\ref{non_rel:thm:non_rel_limit}). The names »semi-« and »non-relativistic« stem from the observation that states whose non-relativistic momentum is $\order(1)$, their semi-relativistic momentum is small, namely $\order(\nicefrac{1}{c})$. This is consistent with the observation that for $\nicefrac{1}{c}$ small enough, the relativistic kinetic energy can be approximated by the non-relativistic kinetic energy, 
\begin{align*}
	\sqrt{m^2 + \bigl ( \nicefrac{\xi}{c} \bigr )^2} = m + \frac{1}{c^2} \frac{1}{2m} \xi^2 + \order(\nicefrac{1}{c^4}) 
	. 
\end{align*}
After making this one choice of scaling, the strategy to prove semi- or non-relativistic dynamics approximate the Dirac dynamics is identical: first, we define a \emph{magnetic} pseudodifferential calculus for position $\mathsf{Q} = \hat{x}$ and semi- or non-relativistic kinetic momentum. Then employing this purpose-built calculus in space-adiabatic perturbation theory gives us recursion formulas to calculate $\Pi$, $U$ and $H_{\mathrm{eff}}$ order-by-order as magnetic $\Psi$DOs. Depending on the choice of scaling, we start with different leading-order terms for $\Pi$, $U$ and $H_{\mathrm{eff}}$. 

Note that using \emph{magnetic} pseudodifferential calculus is \emph{essential} and not optional: technical niceties such as gauge-covariance aside (magnetic $\Psi$DOs written in terms of equivalent gauges $\nabla_x \wedge A = \mathbf{B} = \nabla_x \wedge A'$ are unitarily equivalent, cf.~equation~\eqref{semirel:eqn:gauge-covariance_OpA}), the small parameter $\nicefrac{1}{c}$ appears in front of the vector potential $A$. Thus, usual pseudodifferential calculus which is oblivious to the presence of the magnetic field cannot be used to obtain the correct perturbation expansions. 

In addition to results on approximating the dynamics, we do derive simple spectral results which relate $\sigma(H_{\mathrm{sr \, eff}})$ and $\sigma(H_{\mathrm{nr \, eff}})$ to $\sigma(\widehat{H}_{\mathrm{D}})$ (cf.~Theorems~\ref{semirel:thm:semi-relativistic_limit} and \ref{non_rel:thm:non_rel_limit}). We postpone a more detailed discussion of our main results and how they compare to the ample literature on the subject to Section~\ref{discussion}.

\subsection{Assumptions} % (fold)
\label{intro:assumptions}
Since we employ pseudodifferential methods, we must place regularity assumptions on $\mathbf{B}$ and $V$: 
\begin{assumption}\label{intro:assumption:fields}
	$V$ and the components of $\mathbf{B} = (\mathbf{B}_1,\mathbf{B}_2,\mathbf{B}_3)$ are smooth, bounded functions with bounded derivatives to any order, \ie $V$ and $\mathbf{B}$ are of class $\Cont^{\infty}_{\mathrm{b}}$ for short. 
	
	The components of any vector potential $A = (A_1,A_2,A_3)$ representing $\mathbf{B} = \nabla_x \wedge A$ are smooth polynomially bounded functions with polynomially bounded derivatives to any order, or $A$ is of class $\Cont^{\infty}_{\mathrm{pol}}$ for short. 
\end{assumption}
Note that the assumption on $A$ does not place \emph{any additional} restrictions on the class of admissible magnetic fields $\mathbf{B}$ since any magnetic field of class $\Cont^{\infty}_{\mathrm{b}}$ admits a $\Cont^{\infty}_{\mathrm{pol}}$ vector potential, \eg transversal gauge. Under these assumptions, $\widehat{H}_{\mathrm{D}}$ defines an essentially selfadjoint operator on $\Cont^{\infty}_{\mathrm{c}}(\R^3,\C^4) \cong \Cont^{\infty}_{\mathrm{c}}(\R^3) \otimes \C^4$ (cf.~\cite[Theorem~4.3]{Thaller:Dirac_equation:1992}). 
% subsection Assumptions (end)

\subsection{Structure of the paper} % (fold)
\label{intro:structure}
The paper consists of 6 Sections: first, we discuss the issue of the choice of small parameter and scalings in Section~\ref{scalings}. Then we proceed to explain why the $c \rightarrow \infty$ limit can be seen as an adiabatic limit. Sections~\ref{semirel} and \ref{non_rel} contain the main results of this work, the derivation of semi- and non-relativistic limits of the dynamics, respectively. In addition, we give some spectral results. Lastly, in Section~\ref{discussion}, we compare and contrast our results to previous works. 
% subsection Structure of the paper (end)

\subsection{Acknowledgements} % (fold)
\label{intro:acknowledgements}
The authors would like to thank Harald Lesch, Radu Purice, Herbert Spohn and Stefan Teufel for useful discussions, comments, references and remarks as well as Friedrich Gesztesy for providing several useful references. M.~L. appreciates financial support from the German-Israeli Foundation. M.~F. acknowledges financial support from the DFG cluster of excellence »Origin and Structure of the Universe«.
% subsection Acknowledgements (end)
% section introduction (end)
%!TEX root = /Users/max/Dropbox/research/non- and semi-relativistic limit of the Dirac equation (FL 2008)/non_semi_rel_Dirac.tex
\section{Choice of small parameter and scalings} % (fold)
\label{scalings}
Deciding for a small parameter is crucial since it implicitly contains the physical mechanism that is responsible for the decoupling of electronic and positronic degrees of freedom. Looking at the Dirac hamiltonian, 
\begin{align*}
	\widehat{H}_{\mathrm{D}} = c^2 \, m \beta - i c \eps \nabla_x \cdot \alpha - \mathrm{e} A(\hat{x}) \cdot \alpha + \mathrm{e} V(\hat{x}) \, \id_{\C^4}
	, 
\end{align*}
we see there are four potential parameters, the particle mass $m$, the semiclassical parameter $\eps$, the particle's charge $\mathrm{e}$ and the speed of light $c$.

\subsection{Choice of small parameter} % (fold)
\label{scalings:small_parameter}
Even though we postpone a discussion of the literature to Section~\ref{discussion}, we would like to point out that all four parameters have been used when deriving scaling limits of the Dirac equation.

\subsubsection{Particle mass $m$} % (fold)
\label{scalings:small_parameter:mass}
Foldy and Wouthuysen \cite{Foldy_Wouthuysen:Dirac_theory:1950} have compared typical momenta $\sabs{\xi}$ to the particle mass $m c$, and they order their correction in powers of $\nicefrac{1}{m c}$. The observation that the relativistic kinetic energy 
\begin{align*}
	\sqrt{m^2 c^4 + \xi^2 c^2} = m c^2 \, \sqrt{1 + \bigl ( \nicefrac{\xi}{m c} \bigr )^2} 
	= m c^2 + \tfrac{1}{2m} \xi^2 + \order \bigl ( \nicefrac{\abs{\xi}^4}{m^3 c^2} \bigr ) 
\end{align*}
reduces to the \emph{non}-relativistic kinetic energy if $\abs{\xi} \ll m c$ indicates why Foldy and Wout\-huysen obtain the \emph{non}-relativistic Pauli hamiltonian including higher-order corrections. 
% subsubsection Particle mass $m$ (end)

\subsubsection{Charge $\mathrm{e}$} % (fold)
\label{scalings:small_parameter:charge}
In the absence of magnetic fields, the Douglas-Kroll-Heß block diagonalization method (cf.~\cite{Douglas_Kroll:QED_corrections_helium:1974,Hess:revision_DK_transform:1986,Hess_Jansen:revision_DK_transform:1989,Reiher_Wolf:exact_decoupling_1:2004,Siedentop_Stockmeyer:DKH_method_convergence:2006,Reiher:relativistic_DK_theory:2012} and the discussion in Section~\ref{discussion:block_diagonalization}) uses the charge $\mathrm{e}$ as small parameter and one obtains a block diagonalized hamiltonian in the limit of small coupling to the electric field. 
% subsubsection Charge $e$ (end)

\subsubsection{Semiclassical parameter $\eps$} % (fold)
\label{scalings:small_parameter:eps}
A different approach was taken by other authors (\eg \cite{Brummelhuis_Nourrigat:scattering_Dirac:1999,Spohn:semiclassics_Dirac:2000,PST:sapt:2002}): they used $\eps$ as a small parameter which quantifies a separation of spatial scales. The slow variation of the potentials $A$ and $V$ implies the corresponding fields are weak. In this sense, results obtained from considering the $\eps \rightarrow 0$ limit are not a semi-relativistic limits as the effective hamiltonian (see \eg \cite[equation~(4.9)]{Teufel:adiabatic_perturbation_theory:2003}) may indicate, but a \emph{weak field limit.} 
% subsubsection Semiclassical parameter $\eps$ (end)

\subsubsection{Speed of light $c$} % (fold)
\label{scalings:small_parameter:c}
We will use $\nicefrac{1}{c}$ as small parameter; by far this seems to be the most common choice in other contributions (see \eg \cite{Hunziker:non_rel_limit_Dirac:1975,Grigore_Nenciu_Purice:non_rel_limit_Dirac:1989,Thaller:Dirac_equation:1992}). Physically, the most natural criterion for semi- or non-relativistic behavior is a comparison of energies: if the relativistic kinetic energy, defined as the difference between total relativistic energy and rest energy in a given inertial system, 
\begin{align}
	E_{\mathrm{kin}}(v) := E_{\mathrm{tot}}(v) - E_0 
	= \bigl ( (1 - \nicefrac{v^2}{c^2})^{- \nicefrac{1}{2}} - 1 \bigr ) \, m c^2 
	= \tfrac{1}{2} m v^2 + \orderc{2} 
	, 
	\label{scalings:eqn:energy_comparison}
\end{align}
is smaller than the energy necessary for pair creation, $2 E_0 = 2 m c^2$, then we expect that the Dirac equation describes the physics accurately. In terms of velocities, $\nicefrac{E_{\mathrm{kin}}}{E_0} < 1$ implies $\nicefrac{v}{c} < \nicefrac{\sqrt{3}}{2}$; for energies close to and beyond this critical value, a quantum field theoretical model is necessary. Equation~\eqref{scalings:eqn:energy_comparison} suggests we may equally well use $\nicefrac{v}{c}$ as small parameter. It is \emph{dimensionless} and \emph{relates two inherent physical scales}. Here, those scales are typical speeds $v$ compared to the speed of light $c$. \emph{To make the notation easier on the eyes and our results more easily comparable to the literature, we shall use $\nicefrac{1}{c}$ instead of $\nicefrac{v}{c}$.} 
\medskip

\noindent
The original Dirac hamiltonian, equation~\eqref{intro:eqn:H_D_original}, suggests that typical energies are of order $\mathcal{O}(c^2)$; by rescaling $\widehat{H}_{\mathrm{D}}$ with $\nicefrac{1}{c^2}$ (which is equivalent to rescaling time), we make the leading-order term $\order(1)$: 
\begin{align*}
	\Hd := \tfrac{1}{c^2} \widehat{H}_{\mathrm{D}} 
	&= m \beta + ( - i \tfrac{\eps}{c} \nabla_x ) \cdot \alpha - \tfrac{1}{c^2} A(\hat{x}) \cdot \alpha + \tfrac{1}{c^2} V(\hat{x})
	\\
	&= \Bigl ( m \beta + \bigl ( - i \tfrac{\eps}{c} \nabla_x - \tfrac{1}{c^2} A(\hat{x}) \bigr ) \cdot \alpha \Bigr ) + \tfrac{1}{c^2} V(\hat{x})
	\\
	&= m \beta + \tfrac{1}{c} \bigl ( - i \eps \nabla_x - \tfrac{1}{c} A(\hat{x}) \bigr ) \cdot \alpha + \tfrac{1}{c^2} V(\hat{x}) 
\end{align*}
We have also absorbed the charge $\mathrm{e}$ into the electromagnetic potentials for simplicity. Then the Dirac hamiltonian $\Hd$ suggests \emph{two} natural candidates for the \emph{kinetic} momentum operator: the first choice, 
\begin{align}
	\Psr^A &= - i \tfrac{\eps}{c} \nabla_x - \tfrac{1}{c^2} A(\hat{x}) 
	, 
	\label{scaling:eqn:semi_rel_scaling}
\end{align}
suggests that $\Psr^A \cdot \alpha$ may be of the \emph{same} order of magnitude or even larger than $m \beta$. Or alternatively, we define 
\begin{align}
	\Pnr^A &= - i \eps \nabla_x - \tfrac{1}{c} A(\hat{x}) 
	= c \Psr^A 
	\label{scaling:eqn:non_rel_scaling}
\end{align}
as the kinetic momentum operator to indicate that we are interested in »smaller« momenta and thus, smaller velocities. Arguably, these two choices are the only \emph{natural} scalings suggested by the Dirac equation. 
% subsubsection Speed of light $c$ (end)
% subsection Choice of small parameter (end)

\subsection{Scalings} % (fold)
\label{scalings:scalings}
\subsubsection{Non-relativistic scaling} % (fold)
\label{scalings:scalings:non_rel}
In the non-relativistic scaling where the kinetic momentum operator is given by $\Pnr^A = - i \eps \nabla_x - A(\hat{x})$, we can write the Dirac hamiltonian 
\begin{align*}
	\Hd = \Hnr(\Pnr^A,\Qe) &:= \id_{L^2(\R^3)} \otimes H_{\mathrm{nr} \, 0} + \tfrac{1}{c} H_{\mathrm{nr} \, 1}(\Pnr^A) + \tfrac{1}{c^2} H_{\mathrm{nr} \, 2}(\Qe) 
	\\
	&:= \id_{L^2(\R^3)} \otimes m \beta + \tfrac{1}{c} \Pnr^A \cdot \alpha + \tfrac{1}{c^2} V(\Qe)
\end{align*}
as the sum of three terms. The split implies $H_{\mathrm{nr} \, 1}$ is »small« compared to $H_{\mathrm{nr} \, 0}$. Mathematically speaking however, this is not true, as $H_{\mathrm{nr} \, 1}(\Pnr^A)$ is an unbounded operator and we will have to introduce an energy cutoff later on (cf.~Section~\ref{non_rel:energy_regularization}). 

The leading-order term $H_{\mathrm{nr} \, 0} = m \beta$ is already diagonal so that 
\begin{align}
	\Pi_{\mathrm{nr} \, 0} = \id_{L^2(\R^3)} \otimes \pi_{\mathrm{nr} \, 0} 
	= \left (
	\begin{matrix}
		\id_{\C^2} & 0 \\
		0 & 0 \\
	\end{matrix}
	\right )
\end{align}
projects onto the electronic states, \ie the subspace associated to the positive eigenvalue $m \in \sigma \bigl ( \id_{L^2(\R^3)} \otimes H_{\mathrm{nr} \, 0} \bigr ) = \{ \pm m \}$. 
% subsubsection Non-relativistic scaling (end)

\subsubsection{Semi-relativistic scaling} % (fold)
\label{scalings:scalings:semirel}
Expressed in semi-relativistic units where $\Psr^A = - i \tfrac{\eps}{c} \nabla_x - \tfrac{1}{c^2} A(\hat{x})$, the Dirac operator can be written as 
\begin{align*}
	\Hd = \Hsr(\Psr^A,\Qe) &:= H_{\mathrm{sr} \, 0}(\Psr^A) + \tfrac{1}{c^2} H_{\mathrm{sr} \, 2}(\Qe) 
	\\
	&:= \bigl ( m \beta + \Psr^A \cdot \alpha \bigr ) + \tfrac{1}{c^2} V(\Qe) 
	. 
\end{align*}
Let us neglect the electromagnetic field for the moment, \ie we set $A = 0$ and $V = 0$. Then $H_{\mathrm{sr} \, 0} \bigl ( - i \tfrac{\eps}{c} \nabla_x \bigr )$ diagonalizes to 
\begin{align}
	h_{\mathrm{sr} \, 0} \bigl ( - i \tfrac{\eps}{c} \nabla_x \bigr ) :&= u_{\mathrm{sr} \, 0} \bigl ( - i \tfrac{\eps}{c} \nabla_x \bigr ) \, H_{\mathrm{sr} \, 0} \bigl ( - i \tfrac{\eps}{c} \nabla_x \bigr ) \, u_{\mathrm{sr} \, 0} \bigl ( - i \tfrac{\eps}{c} \nabla_x \bigr )^* 
	\notag \\
	&= \sqrt{m^2 + \bigl ( - i \tfrac{\eps}{c} \nabla_x \bigr )^2} \, \beta 
	=: E \bigl ( - i \tfrac{\eps}{c} \nabla_x \bigr ) \, \beta 
	\label{scaling:eqn:h_sr_0}
\end{align}
via the unitary 
\begin{align}
	u_{\mathrm{sr} \, 0}(\xi) &= \frac{1}{\sqrt{2E(E+m)}} \bigl ( (E+m) \id_{\C^4} - (\xi \cdot \alpha) \beta \bigr ) 
	\label{scaling:eqn:u0sr}
\end{align}
The eigenvalues $\pm \sqrt{m^2 + \xi^2}$ are both two-fold spin-degenerate and the positive (negative) energy eigenspace corresponds to the electronic (positronic) subspace. The projection onto the electronic states is $\pi_{\mathrm{sr} \, 0} \bigl ( - i \tfrac{\eps}{c} \nabla_x \bigr )$ with
\begin{align}
	\pi_{\mathrm{sr} \, 0}(\xi) &= \frac{1}{2} \left ( \id_{\C^4} + \frac{1}{E(\xi)} H_{\mathrm{sr} \, 0}(\xi) \right ) 
	. 
	\label{scaling:eqn:pi0sr}
\end{align}
The fact that both, $u_{\mathrm{sr} \, 0}$ and $\pi_{\mathrm{sr} \, 0}$ depend only on momentum will lead to crucial simplifications in the derivation. 
% subsubsection Semi-relativistic scaling (end)
% subsection Scalings (end)
%
\begin{remark}
	To simplify notation, we will drop the indices ${\;}_{\mathrm{sr}}$ and ${\;}_{\mathrm{nr}}$ whenever there is no risk of confusion. 
\end{remark}
% 
% section general_structure_of_adiabatic_problems_and_scalings (end)
%!TEX root = /Users/max/Dropbox/research/non- and semi-relativistic limit of the Dirac equation (FL 2008)/non_semi_rel_Dirac.tex
\section{The $\nicefrac{1}{c} \rightarrow 0$ limit as an adiabatic limit} % (fold)
\label{adiabatic_limit}
The core of this novel approach is to consider the $\nicefrac{1}{c} \rightarrow 0$ limit of the Dirac equation as an \emph{adiabatic limit.} Whether we obtain the semi- or non-relativistic limit depends solely on the choice of scaling, the rest of the derivation is essentially the same. 

The Dirac equation has three features, the so-called \emph{adiabatic trinity}, shared by all adiabatic systems: 
\begin{enumerate}[(i)]
	\item A distinction between \emph{slow and fast degrees of freedom}, \ie a decomposition of the original Hilbert space the hamiltonian acts on into $\mathcal{H} \cong \Hslow \otimes \Hfast$. Here, the fast Hilbert space is spanned by the electronic and the positronic state, $\Hfast \cong \C^2$. The slow Hilbert space is that of a \emph{non-}relativistic spin-$\nicefrac{1}{2}$ particle, $\Hslow \cong L^2(\R^3,\C^2)$. 
	\item A \emph{small}, dimensionless \emph{parameter} that quantifies the separation of scales. If $v$ is a typical velocity of the particle, we expect that no electron-positron pairs are created as long as $\nicefrac{v}{c} \ll 1$. However, for notational simplicity, we use $\nicefrac{1}{c}$ as small parameter. 
	\item There exists a \emph{relevant part of the spectrum} of the \emph{unperturbed} operator, separated by a \emph{gap} from the remainder. If we consider the field-free case, then $H_0 \bigl (- i \frac{\eps}{c} \nabla_x \bigr )$ fibers via the Fourier transform and the spectrum of each fiber hamiltonian is given by 
	\begin{align*}
		\spec \bigl ( H_0(\xi) \bigr ) = \Bigl \{ \pm \sqrt{m^2 + \xi^2} \Bigr \} 
		. 
	\end{align*}
	We are interested in the electronic subspace, \ie the states associated to $\specrel(x,\xi) := \bigl \{ \sqrt{m^2 + \xi^2} \bigr \}$ -- which is separated by a gap of size 
	\begin{align*}
		2 \sqrt{m^2 + \xi^2} \geq 2m 
	\end{align*}
	from the positronic subspace. This ensures that even in the perturbed case, transitions between electronic and positronic states are exponentially suppressed. 
\end{enumerate}
Put into the form of a commutative diagram, the unperturbed dynamics can be written in the original and »Foldy-Wouthuysen« representation (left and right column), before and after projecting onto the invariant electronic subspace (upper and lower row): 
\begin{align}
	\bfig
		\node L2C4(0,0)[L^2(\R^3,\C^4)]
		\node PiL2C4(0,-600)[\pi_0 \bigl (- i \frac{\eps}{c} \nabla_x \bigr ) \bigl ( L^2(\R^3,\C^4) \bigr )]
		\node L2C4FW(1200,0)[L^2(\R^3,\C^2) \otimes \C^2]
		\node L2C2FW(1200,-600)[L^2(\R^3,\C^2)]
		\arrow[L2C4`PiL2C4;\pi_0(- i \frac{\eps}{c} \nabla_x)]
		\arrow[L2C4`L2C4FW;u_0(- i \frac{\eps}{c} \nabla_x)]
		\arrow[L2C4FW`L2C2FW;\Piref]
		% \arrow[L2C4FW`L2C2FW;\id_{L^2(\R^3,\C^2)} \otimes \piref]
		\arrow/-->/[PiL2C4`L2C2FW;]
		\Loop(0,0){L^2(\R^3,\C^4)}(ur,ul)_{e^{-i t H_0(- i \frac{\eps}{c} \nabla_x)}} 
		\Loop(1200,0){L^2(\R^3,\C^2) \otimes \C^2}(ur,ul)_{e^{-i t E(- i \frac{\eps}{c} \nabla_x) \beta}} 
		\Loop(1200,-600){L^2(\R^3,\C^2)}(dr,dl)^{e^{-i t E(- i \frac{\eps}{c} \nabla_x)}} 
	\efig
	\label{adiabatic_limit:diagram:unperturbed}
\end{align}
In the Foldy-Wouthuysen representation, the free Dirac hamiltonian is diagonal, 
\begin{align*}
	h_0 \bigl (- i \tfrac{\eps}{c} \nabla_x \bigr ) := u_0 \bigl ( - i \tfrac{\eps}{c} \nabla_x \bigr ) \, H_0 \bigl ( - i \tfrac{\eps}{c} \nabla_x \bigr ) \, {u_0}^{\ast} \bigl ( - i \tfrac{\eps}{c} \nabla_x \bigr ) = E \bigl ( - i \tfrac{\eps}{c} \nabla_x \bigr ) \, \beta 
	. 
\end{align*}
In the free case, we are able to describe the electronic and positronic dynamics separately, because the electronic and positronic subspaces are invariant under the dynamics of $H_0(- i \frac{\eps}{c} \nabla_x)$. For the electronic subspace, the \emph{effective hamiltonian} is given by 
\begin{align}
	h_{\mathrm{eff} \, 0} \bigl ( - i \tfrac{\eps}{c} \nabla_x \bigr ) := \Piref \, h_0 \bigl (- i \tfrac{\eps}{c} \nabla_x \bigr ) \, \Piref = E \bigl ( - i \tfrac{\eps}{c} \nabla_x \bigr ) 
\end{align}
where the reference projection $\Piref := \id_{L^2(\R^3)} \otimes \piref$, 
\begin{align*}
	\piref = \left (
	\begin{matrix}
		\id_{\C^2} & 0 \\
		0 & 0 \\
	\end{matrix}
	\right )
	, 
\end{align*}
maps onto the electronic subspace in the Foldy-Wouthuysen representation. The evolution generated by the effective hamiltonian accurately describes the dynamics for electronic states, 
\begin{align}
	\Bigl ( e^{-i t H_0(- i \frac{\eps}{c} \nabla_x)} - {u_0}^{\ast} \bigl (- i \tfrac{\eps}{c} \nabla_x \bigr ) \, e^{-i t E(- i \frac{\eps}{c} \nabla_x)} \, u_0 \bigl (- i \tfrac{\eps}{c} \nabla_x \bigr ) \Bigr ) \, \pi_0 \bigl (- i \tfrac{\eps}{c} \nabla_x \bigr ) = 0 
	. 
	\label{adiabatic_limit:eqn:dynamics_unperturbed}
\end{align}
Hence, we are able to relate the dynamics in the upper-left corner of Diagram~\eqref{adiabatic_limit:diagram:unperturbed} with the reduced, effective dynamics in the lower-right corner. This reduction is made possible, because $H_0 \bigl (- i \frac{\eps}{c} \nabla_x \bigr )$ and $\pi_0 \bigl (- i \frac{\eps}{c} \nabla_x \bigr )$ commute, 
\begin{align*}
	\Bigl [ H_0 \bigl (- i \tfrac{\eps}{c} \nabla_x \bigr ) \; , \, \pi_0 \bigl (- i \tfrac{\eps}{c} \nabla_x \bigr ) \Bigr ] = 0
	, 
\end{align*}
and hence the electronic subspace is invariant under the unperturbed dynamics. 
\medskip

\noindent
If we switch on the electromagnetic perturbation, this is no longer true, the commutator of $\Hd$ and $\pi_0 \bigl (- i \frac{\eps}{c} \nabla_x \bigr )$ is of order $\orderc{3}$. The immediate question is whether we can generalize diagram \eqref{adiabatic_limit:diagram:unperturbed} and equation~\eqref{adiabatic_limit:eqn:dynamics_unperturbed} by replacing $\pi_0 \bigl ( - i \tfrac{\eps}{c} \nabla_x \bigr )$ and $u_0 \bigl ( - i \tfrac{\eps}{c} \nabla_x \bigr )$ with some generalized projection $\Pi$ and a generalized unitary $U$ such that 
\begin{align}
	\bfig
		\node L2C4(0,0)[L^2(\R^3,\C^4)]
		\node PiL2C4(0,-600)[\Pi \bigl ( L^2(\R^3,\C^4) \bigr )]
		\node L2C4FW(1200,0)[L^2(\R^3,\C^2) \otimes \C^2]
		\node L2C2FW(1200,-600)[L^2(\R^3,\C^2)]
		\arrow[L2C4`PiL2C4;\Pi]
		\arrow[L2C4`L2C4FW;U]
		\arrow[L2C4FW`L2C2FW;\Piref]
		\arrow/-->/[PiL2C4`L2C2FW;]
		\Loop(0,0){L^2(\R^3,\C^4)}(ur,ul)_{e^{-i t \Hd}} 
		\Loop(1200,0){L^2(\R^3,\C^2) \otimes \C^2}(ur,ul)_{e^{-i t \hat{h}}} 
		\Loop(1200,-600){L^2(\R^3,\C^2)}(dr,dl)^{e^{-i t \hat{h}_{\mathrm{eff}}}} 
	\efig
	\label{adiabatic_limit:diagram:perturbed}
\end{align}
holds. \emph{If} these operators $\Pi$ and $U$ exist, we require them to be an \emph{orthogonal projection} and a \emph{unitary} which commute with the full perturbed Hamiltonian $\Hd$ and intertwine $\Pi$ and $\Piref$ up to arbitrarily small error in norm in $\nicefrac{1}{c}$, \ie 
\begin{align*}
	&{\Pi}^2 = \Pi, \; {\Pi}^{\ast} = \Pi 
	&& \bigl [ \Hd , \Pi \bigr ] = 0 \\
	&{U}^{\ast} \, U = \id_{L^2(\R^3,\C^4)} , \; U \, {U}^{\ast} = \id_{L^2(\R^3,\C^2) \otimes \C^2} 
	&& U \, \Pi \, {U}^{\ast} = \Piref + \order_{\norm{\cdot}}(\nicefrac{1}{c^{\infty}}) 
\end{align*}
Because of the last property $U$, is called \emph{intertwiner}. The idea of Panati, Spohn and Teufel \cite{PST:sapt:2002} was to use (ordinary) Weyl calculus to obtain the diagonalized and effective hamiltonians $\hat{h}$ and $\hat{h}_{\mathrm{eff}} = \Piref \, \hat{h} \, \Piref$. We will adapt their technique to use a version of \emph{magnetic} Weyl calculus that is tailored to the problem. 
% section general_structure_of_adiabatic_problems_and_scalings (end)
Semi-relativistic%!TEX root = /Users/max/Dropbox/research/non- and semi-relativistic limit of the Dirac equation (FL 2008)/non_semi_rel_Dirac.tex
\section{Semi-relativistic limit} % (fold)
\label{semirel}
We will adapt space-adiabatic perturbation theory \cite{Teufel:adiabatic_perturbation_theory:2003} to derive effective dynamics for electronic initial states in the sense of Diagram~\eqref{adiabatic_limit:diagram:perturbed}. The core idea of space-adiabatic perturbation theory is to use pseudodifferential calculus to write $\Pi = \Op(\pi)$ and $U = \Op(u) + \order_{\norm{\cdot}}(\nicefrac{1}{c^{\infty}})$ as quantization of matrix-valued functions $\pi$ and $u$ which satisfy 
\begin{align}
	\pi \Weyl \pi - \pi &= 0 
	, 
	&&
	\bigl [ \Hsr , \pi \bigr ]_{\Weyl} = 0 
	,
	\label{semirel:eqn:4_defect_equations_symbols}
	\\
	u \Weyl u^* - 1 &= \order(\nicefrac{1}{c^{\infty}}) 
	, \; 
	u^* \Weyl u - 1 = \order(\nicefrac{1}{c^{\infty}}) 
	, 
	&&
	u \Weyl \pi \Weyl u^* - \piref = \order(\nicefrac{1}{c^{\infty}}) 
	. 
	\notag 
\end{align}
The symbol of the \emph{superadiabatic projection} $\pi = \pi_0 + \order(\nicefrac{1}{c})$ and \emph{intertwining unitary} $u = u_0 + \order(\nicefrac{1}{c})$ can be constructed recursively using an asymptotic expansion of the Moyal product $\Weyl$. 

The key ingredient to showing the semi- and non-relativistic limit is to replace standard pseudodifferential theory (see \eg \cite{Folland:harmonic_analysis_hase_space:1989,Hoermander:WeylCalculus:1979}) with a pseudodifferential calculus tailored to the problem.

\subsection{Semi-relativistic pseudodifferential calculus} % (fold)
\label{semirel:mag_PsiDO}
We will use a \emph{magnetic} pseudodifferential calculus for the »building block operators« position and kinetic momentum in semi-relativistic scaling (see Section~\ref{scalings:scalings:semirel}), 
\begin{align}
	\Psr^A &= - i \tfrac{\eps}{c} \nabla_x - \tfrac{1}{c^2} A(\hat{x})
	, 
	\label{semirel:eqn:building_block_ops}
	\\
	\Qsr &= \hat{x} 
	, 
\end{align}
where $A$ is some vector potential associated to the magnetic field $\mathbf{B} = \nabla_x \wedge A$. Associated to these building block operators, there is a quantization $\Opsr^A$, a magnetic Moyal product $\magWsr$ and a magnetic Wigner transform $\mathcal{W}_{\mathrm{sr}}^A$. We refer the interested reader to \cite{Mantoiu_Purice:magnetic_Weyl_calculus:2004} for the basic idea of magnetic $\Psi$DOs and to \cite{Iftimie_Mantiou_Purice:magnetic_psido:2006,Iftimie_Mantoiu_Purice:commutator_criteria:2008,Lein:two_parameter_asymptotics:2008,Lein:progress_magWQ:2010} and references therein for further results. The quantization $\Opsr^A$ depends explicitly on the vector potential $A$ while the formula for $f \magWsr g$ (cf.~\cite[equation~(2.5)]{Lein:two_parameter_asymptotics:2008} where $\eps$ needs to be replaced with $\tfrac{\eps}{c}$ and $\lambda$ by $\tfrac{1}{c^2}$) contains the magnetic \emph{field} $\mathbf{B} = \nabla_x \wedge A$. To ensure that the product of two Hörmander symbols is again a Hörmander symbol (see equation~\eqref{semirel:eqn:Hoermander_symbols} below for a definition), we will need to place Assumption~\ref{intro:assumption:fields} on $\mathbf{B}$ and $A$. Note that the assumption on $A$ does not place \emph{any additional} restrictions on the class of admissible magnetic fields $\mathbf{B}$ since any magnetic field of class $\Cont^{\infty}_{\mathrm{b}}$ admits a $\Cont^{\infty}_{\mathrm{pol}}$ vector potential, \eg transversal gauge. The restriction on $A$ is just necessary to define $\Opsr^A$ via a duality construction. 

Under these assumptions on $\mathbf{B}$ and $A$, we can proceed to define \emph{semi-relativistic pseudodifferential operators:} for a suitable matrix-valued function $h : T^* \R^3 \longrightarrow \mathrm{Mat}_{\C}(4)$, its semi-relativistic quantization is defined by 
\begin{align}
	\Opsr^A(h) := \frac{1}{(2\pi)^3} \int_{T^* \R^3} \negmedspace \negmedspace \negmedspace \dd x \, \dd \xi \, ( \Fourier h )(x,\xi) \, e^{-i (\xi \cdot \Qsr - x \cdot \Psr^A)} \otimes \id_{\C^4}
	\label{semirel:eqn:OpsrA}
\end{align}
where 
\begin{align*}
	( \Fourier h )(x,\xi) := \frac{1}{(2\pi)^3} \int_{T^* \R^3} \negmedspace \negmedspace \negmedspace \dd x' \, \dd \xi' \, e^{i (\xi \cdot x' - x \cdot \xi')} \, h(x',\xi')
\end{align*}
is the symplectic Fourier transform of $h$. Obviously these integrals need to be interpreted appropriately whenever the integrands are not $L^1$, see \eg \cite[Section~V]{Mantoiu_Purice:magnetic_Weyl_calculus:2004} for details. 

One advantage of magnetic pseudodifferential operators is \emph{gauge covariance,} \ie if $A$ and $A' = A + \eps \nabla_x \chi$ are equivalent gauges related by $\chi$, then $\Opsr^A(h)$ and 
\begin{align}
	\Opsr^{A + \eps \nabla_x \chi}(h) = e^{+ \frac{i}{c^2} \chi(\Qsr)} \, \Opsr^A(h) \, e^{- \frac{i}{c^2} \chi(\Qsr)}
	\label{semirel:eqn:gauge-covariance_OpA}
\end{align}
are unitarily equivalent operators. Note that this is \emph{false} if one uses regular Weyl calculus and minimal substitution (see \eg \cite[Section~2.2]{Lein:progress_magWQ:2010}). This implies the product $\magWsr$ implicitly defined through 
\begin{align*}
	\Opsr^A \bigl ( f \magWsr g \bigr ) = \Opsr^A(f) \, \Opsr^A(g) 
\end{align*}
depends on $\mathbf{B}$ rather than $A$. For two Hörmander symbols $f$ and $g$, their product $f \magWsr g$ defined as an oscillatory integral yields another Hörmander symbol \cite[Theorem~2.2]{Iftimie_Mantiou_Purice:magnetic_psido:2006}. This extends to matrix-valued symbols: if $m \in \R$ and $\rho \in [0,1]$, we define the topological vector space of Hörmander symbols of order $m$ and weight $\rho$ as 
\begin{align}
	\Hoermr{m}{\rho} := \Bigl \{ f \in \Cont^{\infty} \bigl ( T^* \R^3 , \mathcal{B}(\C^4) \bigr ) \; \big \vert \; \forall a , \alpha \in \N_0^3  \; \snorm{f}_{m , a \alpha} < \infty \Bigr \} 
	\label{semirel:eqn:Hoermander_symbols} 
	, 
\end{align}
which is equipped with a Fréchet structure generated by the family of seminorms 
\begin{align*}
	\snorm{f}_{m , a \alpha} := \sup_{(x,\xi) \in T^* \R^3} \Bigl ( {\sqrt{1 + \xi^2}}^{\, -m + \abs{\alpha} \rho} \, \bnorm{\partial_x^a \partial_{\xi}^{\alpha}f(x,\xi)}_{\mathcal{B}(\C^4)} \Bigr )
	, 
	&& 
	a , \alpha \in \N_0^3
	. 
\end{align*}
The space of rapidly decaying symbols $\Hoermr{-\infty}{\rho} := \bigcap_{m \in \R} \Hoermr{m}{\rho}$ is defined as the intersection of all Hörmander symbols and equipped with the projective limit topology. 

Many standard results of pseudodifferential theory are also available for \emph{magnetic} pseudodifferential operators, \eg there is a magnetic version of the Caldéron-Vaillancourt theorem \cite[Theorem~3.1]{Iftimie_Mantiou_Purice:magnetic_psido:2006}, the quantization of elliptic real-valued Hörmander symbols of positive order $m > 0$ yields selfadjoint operators on $L^2(\R^d)$ with domain $H^m_A(\R^d)$ \cite[Theorem~4.1]{Iftimie_Mantiou_Purice:magnetic_psido:2006}, and there are Beals- and Bony-type commutator criteria \cite{Iftimie_Mantoiu_Purice:commutator_criteria:2008}. These results again extend to matrix-valued symbols as outlined in \cite[Appendix~A]{Teufel:adiabatic_perturbation_theory:2003} and Appendix~\ref{appendix:mag_PsiDOs:extension}. 

Of particular importance is the two-parameter version of magnetic Weyl calculus developed in \cite{Lein:two_parameter_asymptotics:2008}: if the semiclassical parameter $\eps$ and the coupling constant $\lambda$ are replaced by $\tfrac{\eps}{c}$ and $\tfrac{1}{c^2}$, respectively, we obtain the semi-relativistic pseudodifferential calculus used here. First, we need the notion of symbols which have a »good« expansion in $\nicefrac{1}{c}$; these correspond to the »semiclassical symbols« introduced in \cite[Definition~A.12]{Teufel:adiabatic_perturbation_theory:2003}: 
\begin{defn}[Semi-relativistic symbols]\label{semirel:defn:semi-relativistic_symbol}
	A map $f : [c_0,\infty) \longrightarrow \Hoermr{m}{1}$, $c \mapsto f_c$, is called a semiclassical symbol of order $m \in \R$, that is $f \in \SemiHoer{m}$, if there exists a sequence $\{ f_n \}_{n \in \N_0}$, $f_n \in \Hoer{m - \frac{2}{3} n}$, such that for all $N \in \N_0$, one has 
	\begin{align*}
		c^{N} \biggl ( f_c - \sum_{n = 0}^{N-1} \tfrac{1}{c^n} \, f_n \biggr ) \in \Hoer{m - \frac{1}{3} N}
	\end{align*}
	uniformly in $c \in [c_0,+\infty)$ in the sense that for any $N \in \N_0$ and $a , \alpha \in \N_0^3$, there exists constants $C_{N \, a \alpha} > 0$ such that 
	\begin{align*}
		\norm{f_c - \sum_{n = 0}^{N-1} \tfrac{1}{c^n} \, f_n}_{m - \frac{1}{3} N , a \alpha} \leq \tfrac{1}{c^N} \, C_{N \, a \alpha} 
	\end{align*}
	holds for all $c \in [c_0,\infty)$. 
\end{defn}
While the appearance of the factor $\nicefrac{1}{3}$ may seem odd, the reason will become clear in the proof of Theorem~\ref{semirel:thm:semi-relativistic_limit}; in any case, this is just a technical point. 

The asymptotic expansion of $\magWsr$ in powers of $\nicefrac{1}{c}$ can be obtained using \cite[Theorem~1.1]{Lein:two_parameter_asymptotics:2008}: 
\begin{lem}\label{semirel:lem:asymp_expansion_magW}
	Let $f \in \Hoermr{m_1}{\rho}$, $g \in \Hoermr{m_2}{\rho}$, $m_1 , m_2 \in \R$, $\rho \in (0,1]$, be two Hörmander symbols and assume the components of $\mathbf{B}$ are of class $\Cont^{\infty}_{\mathrm{b}}$. 
	\begin{enumerate}[(i)]
		\item Then $f \magWsr g$ expands asymptotically in $\nicefrac{1}{c} \,$, \ie for each $N \in \N_0$, we can write 
		\begin{align*}
			f \magWsr g &= \sum_{n = 0}^N \tfrac{1}{c^n} \, \bigl ( f \magWsr g \bigr )_{(n)} + \tfrac{1}{c^{N+1}} \, \tilde{R}_N(f,g) 
		\end{align*}
		where all terms are known explicitly. The $j$th term of the expansion
		\begin{align*}
			\bigl ( f \magWsr g \bigr )_{(n)} = \sum_{3k \leq n} \eps^{n - 2k} \, \bigl ( f \magW g \bigr )_{(n-2k,k)}
		\end{align*}
		can be expressed in terms of the $(f \magW g)_{(n,k)}$ from \cite[Theorem~1.1]{Lein:two_parameter_asymptotics:2008}. Furthermore, each term of the expansion is in symbol class 
		\begin{align*}
			\bigl ( f \magWsr g \bigr )_{(n)} \in \Hoermr{m_1 + m_2 - \frac{2}{3} n \rho}{\rho}
		\end{align*}
		and the remainder $\tilde{R}_{N}$ is an element of $\Hoermr{m_1 + m_2 - \frac{2}{3} (N+1) \rho}{\rho}$.
		\item If in addition $f \asymp \sum_{n = 0}^{\infty} \tfrac{1}{c^n} \, f_n$ and $g \asymp \sum_{n = 0}^{\infty} \tfrac{1}{c^n} \, g_n$ admit asymptotic expansions with $f_n \in \Hoermr{m_1 - \frac{2}{3} n \rho}{\rho}$ and $g_n \in \Hoermr{m_2 - \frac{2}{3} n \rho}{\rho}$, $n \in \N_0$, then $f \magWsr g \asymp \sum_{n = 0}^{\infty} \tfrac{1}{c^n} \, h_n$ admits an asymptotic expansion in $\nicefrac{1}{c}$ where 
		\begin{align*}
			h_n = \sum_{k + l + j = n} \bigl ( f_k \magWsr g_l \bigr )_{(j)} \in \Hoermr{m_1 + m_2 - \frac{2}{3} n \rho}{\rho}
			. 
		\end{align*}
		\item If $f \asymp \sum_{n = 0}^{\infty} \tfrac{1}{c^n} \, f_n \in \SemiHoer{m_1}$ and $g \in \SemiHoer{m_2}$ are two relativistic symbols, then $f \magWsr g \asymp \sum_{n = 0}^{\infty} \tfrac{1}{c^n} \, g_n \in \SemiHoer{m_1 + m_2}$ is also a semi-relativistic symbol whose asymptotic resummation is given in (ii). 
	\end{enumerate}
\end{lem}
The proof which can be found in Appendix~\ref{appendix:mag_PsiDOs:expansion_magWsr} is straightforward and mostly book-keeping of symbol classes. 
\begin{remark}\label{semirel:rem:product_two_momentum_functions}
	If $f$ and $g$ are functions of momentum only, then the first non-trivial term of the asymptotic expansion is a purely magnetic term; the first purely magnetic term $\bigl ( f \magW g \bigr )_{(1,1)}$ appears at third order in $\nicefrac{1}{c}$ (cf.~Appendix~\ref{appendix:mag_PsiDOs:asymptotic_expansions}), \ie we have 
	\begin{align*}
		f \magWsr g = f \, g + \orderc{3} 
		. 
	\end{align*}
	This fact will simplify calculations tremendously. 
\end{remark}
While working with asymptotic expansions, we will need two more conventions regarding the use of Landau symbols: we say that a function $f_c \asymp \sum_{n = 0}^{\infty} \tfrac{1}{c^n} \, f_n = \order(\nicefrac{1}{c^{\infty}})$ if and only if $f_c = \order(\nicefrac{1}{c^N})$ holds for all $N \in \N_0$. Secondly, an operator $A_c$ is $\order_{\norm{\cdot}}(\nicefrac{1}{c^n})$, $n \in \N_0 \cup \{ \infty \}$, if $\norm{A_c} = \order(\nicefrac{1}{c^n})$. 
% subsection Magnetic pseudodifferential operators (end)

\subsection{Approximate quantum dynamics} % (fold)
\label{semirel:quantum_dynamics}
Now we will establish the existence of effective dynamics in the sense of Diagram~\eqref{adiabatic_limit:diagram:unperturbed} and give the first few orders of the effective hamiltonian explicitly: 
\begin{thm}[Semi-relativistic limit]\label{semirel:thm:semi-relativistic_limit}
	Let Assumption~\ref{intro:assumption:fields} on $\mathbf{B}$, $V$ and $A$ be satisfied. Then electronic and positronic degrees of freedom decouple to any order in $\nicefrac{1}{c} \ll 1$ in the sense of Diagram~\eqref{adiabatic_limit:diagram:unperturbed}: \ie there exist 
	\begin{enumerate}[(i)]
		\item a projection $\Pi = \Opsr^A(\pi) = \Opsr^A(\pi_0) + \order_{\norm{\cdot}}(\nicefrac{1}{c^3})$, $\pi \in \SemiHoer{0}$, 
		\item an intertwining unitary $U = \Opsr^A(u) + \order_{\norm{\cdot}}(\nicefrac{1}{c^{\infty}}) = \Opsr^A(u_0) + \order_{\norm{\cdot}}(\nicefrac{1}{c^3})$, $u \in \SemiHoer{0}$, 
		\item and an effective hamiltonian which is the semi-relativistic quantization of 
		\begin{align}
			h_{\mathrm{eff}} &= \piref \, u \magWsr H \magWsr u^* \, \piref 
			% \asymp \sum_{j = 0}^{\infty} \tfrac{1}{c^j} \, h_{\mathrm{eff} \, j}
			\notag \\
			&= E + \tfrac{1}{c^2} V - \tfrac{1}{c^3} \tfrac{\eps}{2 E (E+m)} 
				% \Bigl ( 
					\bigl ( (E+m) \, \mathbf{B} - (\nabla_x V \wedge \xi) \bigr ) \cdot \sigma \oplus 0_{\C^2}
				% \Bigr ) 
			+ \order(\nicefrac{1}{c^4}) 
			\in \SemiHoer{1}
		\end{align}
		where $E(\xi) := \sqrt{m^2 + \xi^2}$ and $\sigma_j$, $j = 1 , 2 , 3$ are the Pauli matrices. 
	\end{enumerate}
	These operators satisfy 
	\begin{align}
		\bigl [ \Hd , \Pi \bigr ] = 0 
	\end{align}
	and for electronic initial states, namely those in $\Pi L^2(\R^3,\C^4)$, the full dynamics is approximated by the effective dynamics to any order in $\nicefrac{1}{c}$, 
	\begin{align}
		\Bnorm{\Bigl ( e^{- i t \Hd} -  U^* \, e^{- i t \Opsr^A(h_{\mathrm{eff}})} \, U \Bigr ) \Pi} = \order \bigl ( \nicefrac{(1 + \abs{t})}{c^{\infty}} \bigr )
		. 
	\end{align}
\end{thm}
Since only some technicalities of the proof of \cite[Theorem~3.2]{Teufel:adiabatic_perturbation_theory:2003}  need to be modified, we will content ourselves to give an outline of the proof and focus on the necessary modifications as in \cite{DeNittis_Lein:Bloch_electron:2009}. To improve readability, we will show the existence of $\Pi$, $U$ and $\heff$ separately. 
\begin{lem}\label{semirel:lem:existence_Pi}
	Under the assumptions of Theorem~\ref{semirel:thm:semi-relativistic_limit}, there exists a projection 
	\begin{align*}
		\Pi = \Pi^2 
		= \Opsr^A(\pi) 
	\end{align*}
	which commutes with $\Hd$, $\bigl [ \Hd , \Pi \bigr ] = 0$, and is the quantization of a Moyal projector 
	\begin{align*}
		\pi \asymp \sum_{n = 0}^{\infty} \tfrac{1}{c^n} \, \pi_n \in \SemiHoer{0} 
		. 
	\end{align*}
	Furthermore, the first non-trivial correction of $\pi$ is of third order, \ie $\pi = \pi_0 + \order(\nicefrac{1}{c^3})$. 
\end{lem}
\begin{proof}
	First of all, we note that $\Hd = \Opsr^A(\Hsr)$ is selfadjoint on $H^1_A(\R^3,\C^4)$ by Corollary~\ref{appendix:mag_PsiDOs:cor:selfadjointness_H_D}. Clearly $\Hsr \in \SemiHoer{1}$ holds and the spectral gap between $\specrel(x,\xi)$ and the remainder of $\sigma \bigl ( H_0(x,\xi) \bigr )$ is equal to $2 \sqrt{m^2 + \xi^2}$. Thus $(\mathrm{IG})_1$ (cf.~\cite[p.~74]{Teufel:adiabatic_perturbation_theory:2003}) is satisfied. 
	
	Then for any $(x_0,\xi_0) \in T^* \R^3$ and a suitably chosen contour $\Gamma(x_0,\xi_0)$, the Moyal projection
	\begin{align}
		\pi(x,\xi) := \frac{i}{2\pi} \int_{\Gamma(x_0,\xi_0)} \dd z \, \bigl ( \Hsr - z \bigr )^{(-1)^B_{\mathrm{sr}}}(x,\xi)
		\label{semirel:eqn:pi_exact}
	\end{align}
	can locally be written as a Cauchy integral of the $n$th term of the Moyal resolvent, \ie the tempered distribution which satisfies 
	\begin{align*}
		\bigl ( \Hsr - z \bigr ) \magWsr \bigl ( \Hsr - z \bigr )^{(-1)^B_{\mathrm{sr}}} = 1 
		= \bigl ( \Hsr - z \bigr )^{(-1)^B_{\mathrm{sr}}} \magWsr \bigl ( \Hsr - z \bigr )
		. 
	\end{align*}
	Proposition~\ref{appendix:mag_PsiDOs:prop:inversion} tells us the Moyal resolvent is a symbol of class $\Hoermr{-1}{1}$. From the Gap Condition, we know we can always choose a contour whose circumference does not exceed $C \sqrt{1 + \xi^2}$ for some $C > 0$ independent of $x$ and $\xi$. Hence, $\pi$ as defined in \eqref{semirel:eqn:pi_exact} is in $\Hoermr{0}{1}$. By the magnetic Caldéron--Vaillancourt Theorem \cite[Theorem~3.1]{Iftimie_Mantiou_Purice:magnetic_psido:2006} and our assumptions on $\mathbf{B}$ and $A$, $\Opsr^A(\pi)$ defines a bounded operator on $L^2(\R^3,\C^4)$. 
	
	Using the (Moyal) resolvent identity, we conclude that $\pi$ is really a Moyal \emph{projection} and thus $\Pi^2 = \Pi = \Opsr^A(\pi)$ is a projection in the operator sense. By definition, $\bigl ( \Hsr - z \bigr )^{(-1)^B_{\mathrm{sr}}}$ Moyal commutes with $\Hsr$ which in view of equation~\eqref{semirel:eqn:pi_exact} implies 
	\begin{align*}
		\bigl [ \Hsr , \pi \bigr ]_{\magWsr} = \Hsr \magWsr \pi - \pi \magWsr \Hsr = 0 
		. 
	\end{align*}
	The $\nicefrac{1}{c}$ expansion of $\magWsr$ (Lemma~\ref{semirel:lem:asymp_expansion_magW}) yields asymptotic expansion of the Moyal resolvent, 
	\begin{align*}
		\bigl ( \Hsr - z \bigr )^{(-1)^B_{\mathrm{sr}}} \asymp \sum_{n = 0}^{\infty} \tfrac{1}{c^n} \, R_n(z) 
		, 
	\end{align*}
	where each of the terms $R_n(z) \in \Hoermr{- 1 - \nicefrac{n}{3}}{1}$ is in the correct symbol class. Since the gap increases as $\sqrt{1 + \xi^2}$, we can estimate the seminorms of $\pi_n$ by seminorms of $R_n$ times the length of the contour which is at most $C \sqrt{1 + \xi^2}$, and $\pi_n$ locally defined as a contour integral over the $R_n(z)$ as in \eqref{semirel:eqn:pi_exact}, is an element of $\Hoermr{-\nicefrac{n}{3}}{1}$. 
		
	To compute the terms $\pi_n$, we make straightforward modifications to the proof of  \cite[Lemma~3.8]{Teufel:adiabatic_perturbation_theory:2003}: we note 
	\begin{align*}
		\pi_0 \magWsr \pi_0 - \pi_0 = \order(\nicefrac{1}{c^3}) \in \Hoermr{-2}{1} 
		\subset \Hoermr{-1}{1}
	\end{align*}
	and
	\begin{align*}
		\bigl [ \Hsr , \pi_0 \bigr ]_{\magWsr} = \order(\nicefrac{1}{c^3}) \in \Hoermr{-1}{1} 
		, 
	\end{align*}
	hold true by Lemma~\ref{semirel:lem:asymp_expansion_magW}, Remark~\ref{semirel:rem:product_two_momentum_functions} and 
	\begin{align*}
		\bigl [ V , \pi_0 \bigr ]_{\magWsr} = \tfrac{\eps}{c} \bigl \{ V , \pi_0 \bigr \} + \order(\nicefrac{1}{c^2}) \in \Hoermr{-1}{1} 
		. 
	\end{align*}
	Hence, $\pi_0$ satisfies the induction assumption in Teufel's proof, and the rest of the proof in \cite{Teufel:adiabatic_perturbation_theory:2003} can be transliterated with obvious modifications. The reader may check that the same recursion relations hold for each $\pi_n$, replacing the $\eps$ expansion of $\Weyl$ with the $\nicefrac{1}{c}$ expansion of $\magWsr$ and using Lemma~\ref{semirel:lem:asymp_expansion_magW}. 
\end{proof}
\begin{remark}
	Note that Proposition~\ref{appendix:mag_PsiDOs:prop:inversion} simplifies Teufel's proof: we no longer need to define $\Pi$ as spectral projection of $\Opsr^A(\pi)$. 
\end{remark}
\begin{lem}\label{semirel:lem:existence_U}
	Under the assumptions of Theorem~\ref{semirel:thm:semi-relativistic_limit}, there exists a unitary operator 
	\begin{align*}
		U = \Opsr^A(u) + \order_{\norm{\cdot}}(\nicefrac{1}{c^{\infty}}) 
	\end{align*}
	such that $U \Pi U^* = \Piref$ and $U$ is $\order_{\norm{\cdot}}(\nicefrac{1}{c^{\infty}})$-close to the quantization of an almost-Moyal unitary 
	\begin{align*}
		u \asymp \sum_{n = 0}^{\infty} \tfrac{1}{c^n} \, u_n \in \SemiHoer{0} 
		. 
	\end{align*}
	Furthermore, the first non-trivial correction of $u$ is of third order, \ie $u = u_0 + \order(\nicefrac{1}{c^3})$. 
\end{lem}
\begin{proof}
	We have to make the necessary modifications to the proof of Lemma~3.15 in \cite{Teufel:adiabatic_perturbation_theory:2003} by replacing standard pseudodifferential operators with semi-relativistic pseudodifferential operators, analogously to the proof of Lemma~\ref{semirel:lem:existence_Pi}. 
	
	Again, the magnetic Caldéron--Vaillancourt theorem \cite[Theorem~3.1]{Iftimie_Mantiou_Purice:magnetic_psido:2006} and our assumptions on $\mathbf{B}$ and $A$ imply the operator $\Opsr^A(u_0)$ is bounded. 
	
	$u_0 \in \Hoermr{0}{1}$ as given by equation~\eqref{scaling:eqn:u0sr} is in the correct symbol class and satisfies 	
	\begin{align*}
		u_0 \magWsr {u_0}^* - 1 &= \order(\nicefrac{1}{c^3}) \in \Hoermr{-2}{1} \subset \Hoermr{-1}{1}
		\\
		{u_0}^* \magWsr u_0 - 1 &= \order(\nicefrac{1}{c^3}) \in \Hoermr{-2}{1} \subset \Hoermr{-1}{1}
	\end{align*}
	as well as 
	\begin{align*}
		u_0 \magWsr \pi \magWsr {u_0}^* - \piref &= u_0 \pi_0 {u_0}^* - \piref + \order(\nicefrac{1}{c^3}) 
		= \order(\nicefrac{1}{c^3}) \in \Hoermr{-1}{1} 
		, 
	\end{align*}
	and we may proceed by induction as in \cite{Teufel:adiabatic_perturbation_theory:2003}. In particular, the above equations imply $u = u_0 + \order(\nicefrac{1}{c^3})$. 
	
	Lastly, the true unitary $U$ is obtained from $u$ and $\pi$ through the Nagy formula (cf.~\cite[pp.~87--88]{Teufel:adiabatic_perturbation_theory:2003}). 
\end{proof}
\begin{proof}[Proof of Theorem~\ref{semirel:thm:semi-relativistic_limit}]
	From the proof of Lemma~\ref{semirel:lem:existence_U}, we know that $\Hd$ defines a selfadjoint operator on $H^1_A(\R^3,\C^4)$, and $\Hsr$ and $\specrel(x,\xi) = \bigl \{ \sqrt{m^2 + \xi^2} \bigr \}$ satisfy $(\mathrm{IG})_1$. 
	
	The existence of $\Pi = \Opsr^A(\pi)$ and $U = \Opsr^A(u) + \order(\nicefrac{1}{c^{\infty}})$ has been proven in Lemmas~\ref{semirel:lem:existence_Pi} and \ref{semirel:lem:existence_U}. Hence, the effective hamiltonian 
	\begin{align*}
		\heff := \piref \, u \magWsr H \magWsr u^* \, \piref \in \SemiHoer{1}
	\end{align*}
	exists and is in the correct symbol class by Lemma~\ref{semirel:lem:asymp_expansion_magW}. From Lemma~\ref{appendix:mag_PsiDOs:lem:selfadjointness} we deduce that $\Opsr^A(\heff)$ defines a selfadjoint operator on $H^1_A(\R^3,\C^2) \oplus L^2(\R^3,\C^2)$. 
	
	A standard Duhamel argument shows that $U^* \, e^{- i t \Opsr^A(h)} \, U$ approximates the full dynamics generated by $\Hd$ for states in $\Pi L^2(\R^3,\C^4)$ up to $\order_{\norm{\cdot}}(\nicefrac{1}{c^{\infty}})$ (cf.~\cite[proof of Theorem~3.20]{Teufel:adiabatic_perturbation_theory:2003}). 
	
	Computing $\heff$ is straightforward, but tedious; the details can be found in Section~III and Appendix~E of \cite{Lein:two_parameter_asymptotics:2008}, one only needs to insert powers of $\eps$ in the appropriate places. This concludes the proof. 
\end{proof}
%
% subsection Effective quantum dynamics (end)

\subsection{Spectral results} % (fold)
\label{semirel:spectrum}
Lastly, we would like to show how to infer from the presence of spectrum of $\Hd$ in the vicinity of $E_0 > 0$ the presence of spectrum of $\Opnr^A(\heff)$ in a possibly larger neighborhood of $E_0$ and vice versa. 

For any $k \in \N_0$, we introduce the finite summations 
\begin{align*}
	\Pi^{(k)} &:= \sum_{n = 0}^k \tfrac{1}{c^n} \, \Opsr^A(\pi_n)
	\\
	U^{(k)} &:= \sum_{n = 0}^k \tfrac{1}{c^n} \, \Opsr^A(u_n)
	\\
	H_{\mathrm{eff}}^{(k)} &:= \sum_{n = 0}^k \tfrac{1}{c^n} \, \Opsr^A( h_{\mathrm{eff} \, n})
	. 
\end{align*}
From the very definition of these objects, we know that $\Pi^{(k)} = {\Pi^{(k)}}^2 + \order_{\norm{\cdot}}(\nicefrac{1}{c^{k+1}})$ is an almost-projection and $U^{(k)}$ an almost-unitary. 
\begin{thm}\label{semirel:thm:spectrum}
	For any $k \in \N_0$, the following statements hold true: 
	\begin{enumerate}[(i)]
		\item Let $E_0 \in \sigma(\Hd) \cap [0,+\infty)$. Then for any $\delta > 0$, there exists $\Psi_{\delta} \in L^2(\R^3,\C^4)$, $\snorm{\Psi_{\delta}} = 1$, such that 
		\begin{align*}
			\Bnorm{\bigl ( H_{\mathrm{eff}}^{(k)} - E_0 \bigr ) \, U^{(k)} \, \Pi^{(k)} \, \Psi_{\delta}} < C_k \, \delta + \order(\nicefrac{1}{c^{k+1}})
		\end{align*}
		holds where $C_k$ and the $\order(\nicefrac{1}{c^{k+1}})$ term are independent of $\delta$. 
		\item Similarly, if $E_0 \in \sigma \bigl ( H_{\mathrm{eff}}^{(k)} \bigr ) \cap [0,+\infty)$, then for any $\delta > 0$, there exists $\Psi_{\delta} \in L^2(\R^3,\C^4)$ with $\snorm{\Psi_{\delta}} = 1$ such that 
		\begin{align*}
			\Bnorm{\bigl ( \Hd - E_0 \bigr ) \, {U^{(k)}}^* \, \Pi^{(k)} \, \Psi_{\delta}} < C_k \, \delta + \order(\nicefrac{1}{c^{k+1}})
		\end{align*}
		holds where $C_k$ and the $\order(\nicefrac{1}{c^{k+1}})$ term are independent of $\delta$. 
	\end{enumerate}
\end{thm}
\begin{proof}
	\begin{enumerate}[(i)]
		\item If $E_0 \in \sigma(\Hd) \cap \R^+$ lies in the spectrum, then the Weyl criterion implies that for any $\delta > 0$, we may find a $\Psi_{\delta} \in L^2(\R^3,\C^4)$ of norm $1$ such that 
		\begin{align*}
			\bnorm{\bigl ( \Hd - E_0 \bigr ) \Psi_{\delta}} < \delta 
		\end{align*}
		holds. Using the defining properties of $\Pi^{(k)}$, $U^{(k)}$ and $H_{\mathrm{eff}}^{(k)}$ we obtain 
		\begin{align*}
			\bigl ( H_{\mathrm{eff}}^{(k)} - E_0 \bigr ) U^{(k)} \Pi^{(k)} &= {U^{(k)}}^* \bigl ( \Hd - E_0 \bigr ) \Pi^{(k)} + \order_{\norm{\cdot}}(\nicefrac{1}{c^{k+1}}) 
			\\
			&
			= {U^{(k)}}^* \Pi^{(k)} \bigl ( \Hd - E_0 \bigr ) + \order_{\norm{\cdot}}(\nicefrac{1}{c^{k+1}})
		\end{align*}
		and hence 
		\begin{align*}
			&\Bnorm{\bigl ( H_{\mathrm{eff}}^{(k)} - E_0 \bigr ) U^{(k)} \Pi^{(k)} \Psi_{\delta}} 
			< \\
			&\qquad \qquad 
			< \bnorm{U^{(k)}} \, \bnorm{\Pi^{(k)}} \, \delta + \bnorm{\order_{\norm{\cdot}}(\nicefrac{1}{c^{k+1}})} \, \bnorm{\Psi_{\delta}}
			= C_k \, \delta + \order(\nicefrac{1}{c^{k+1}}) 
			. 
		\end{align*}
		\item The proof is completely analogous to (i). 
	\end{enumerate}
\end{proof}
%
% subsection Spectral results (end)
% section semi-relativistic limit (end)
%!TEX root = /Users/max/Dropbox/research/non- and semi-relativistic limit of the Dirac equation (FL 2008)/ahp_non_semi_rel_Dirac.tex
\section{Non-relativistic limit} % (fold)
\label{non_rel}
The non-relativistic limit turns out to be more technically involved than the semi-relativistic limit. Here, the sub-leading term to the effective hamiltonian 
\begin{align*}
	\widehat{h}_{\mathrm{eff}} = m \, \Piref + \frac{1}{c^2} \left ( \frac{1}{2m} \bigl ( - i \eps \nabla_x - \tfrac{1}{c} A(\hat{x}) \bigr )^2 + V(\hat{x}) \right ) \, \Piref + \order(\nicefrac{1}{c^3})
\end{align*}
is more singular than both, $\widehat{h}_{\mathrm{eff} \, 0} = m \, \Piref$ (which is bounded) and $\Hd$. It turns out that $\widehat{h}_{\mathrm{eff} \, 4}$ contains a term proportional to the quantization of $\abs{\xi}^4$ which dominates the behavior of $\heff$ for large $\xi$. To control this, we must look at states whose energy is finite. Before we proceed, let us adapt Section~\ref{adiabatic_limit} to the present context. 

\subsection{Non-relativistic limit as an adiabatic limit} % (fold)
\label{non_rel:adiabatic_limit}
The first two items of the adiabatic trinity, the splitting into slow and fast degrees of freedom, $\Hslow \otimes \Hfast = L^2(\R^3,\C^2) \otimes \C^2$, as well as the small parameter $\nicefrac{v}{c} \ll 1$, are the same as in the semi-relativistic case. It remains to consider the relevant part of the spectrum: the leading-order term of the hamiltonian symbol in the non-relativistic scaling, 
\begin{align*}
	H_{\mathrm{nr} \, 0}(x,\xi) = m \beta 
	= H_{\mathrm{sr} \, 0} \bigl ( x , \nicefrac{\xi}{c} \bigr ) + \order(\nicefrac{1}{c})
	, 
\end{align*}
approximates $H_{\mathrm{sr}}$ for small momenta. Similarly, the spectrum 
\begin{align*}
	\sigma \bigl ( H_{\mathrm{nr} \, 0}(x,\xi) \bigr ) = \{ \pm m \} 
\end{align*}
can be seen as the dominant term of the Taylor expansion of 
\begin{align*}
	\sigma \bigl ( H_{\mathrm{sr} \, 0}(x,\nicefrac{\xi}{c}) \bigr ) = \Bigl \{ \pm \sqrt{m^2 + \bigl ( \nicefrac{\xi}{c} \bigr )^2} \Bigr \} 
	= \Bigl \{ \pm m + \order(\nicefrac{1}{c^2}) \Bigr \} 
\end{align*}
for small $\nicefrac{1}{c}$. Hence, 
\begin{align*}
	\pi_{\mathrm{nr} \, 0} = \left (
	\begin{matrix}
		\id_{\C^2} & 0 \\
		0 & 0 \\
	\end{matrix}
	\right )
\end{align*}
projects onto the relevant band, namely the electronic states $\specrel(x,\xi) = \{ + m \}$. Since $H_{\mathrm{nr} \, 0} = m \beta$ is already diagonal, we may use 
\begin{align*}
	u_{\mathrm{nr} \, 0}(x,\xi) = \id_{\C^4}
\end{align*}
as Moyal unitary which diagonalizes $H_{\mathrm{eff} \, 0}$. 
\begin{remark}
	From this point forward, we shall drop the index ${\;}_{\mathrm{nr}}$ from most of the objects to simplify notation. 
\end{remark}
%
% subsection Non-relativistic limit as an adiabatic limit (end)

\subsection{Non-relativistic pseudodifferential calculus} % (fold)
\label{non_rel:psiDO}
Analogously to Section~\ref{semirel:mag_PsiDO}, we will introduce a magnetic pseudodifferential calculus associated to position and \emph{non-relativistic} kinetic momentum, 
\begin{align*}
	\Pnr^A &= - i \eps \nabla_x - \tfrac{1}{c} A(\hat{x}) 
	\\
	\Qnr &= \hat{x}
	. 
\end{align*}
We emphasize that \emph{$\eps$ need not be small} as we are interested in the weak coupling limit where $\nicefrac{1}{c} \ll 1$. The technical details (\eg assumptions on $\mathbf{B}$ and $A$ as well as the matrix-valued functions involved) are identical to those for semi-relativistic $\Psi$DO, and hence there is no need to repeat them. Now a non-relativistic pseudodifferential operator associated to a suitable $4 \times 4$-matrix-valued function $f$ on $T^* \R^3$ is defined as 
\begin{align*}
	\Opnr^A(f) := \frac{1}{(2\pi)^3} \int_{T^* \R^3} \dd x \, \dd \xi \, (\Fourier_{\sigma} f)(x,\xi) \, e^{-i (\xi \cdot \Qnr - x \cdot \Pnr^A)} \otimes \id_{\C^4} 
	. 
\end{align*}
The non-relativistic Moyal product $\magWnr$ which is implicitly defined through 
\begin{align*}
	\Opnr^A \bigl ( f \magWnr g \bigr ) = \Opnr^A(f) \, \Opnr^A(g)
\end{align*}
has an asymptotic expansion in $\nicefrac{1}{c}$: 
\begin{thm}[Theorem~2.12 of \cite{Lein:two_parameter_asymptotics:2008}]\label{non_rel:thm:properties_magWnr}
	Let $f \in \Hoerr{m_1}$, $g \in \Hoerr{m_2}$, $m_1 , m_2 \in \R$, $\rho \in (0,1]$, and assume the components of $\mathbf{B}$ are $\Cont^{\infty}_{\mathrm{b}}$ functions. Then $f \magWnr g$ has an asymptotic expansion in $\nicefrac{1}{c}$: for any $N \in \N_0$, we can write 
		\begin{align}
			f \magWnr g = \sum_{n = 0}^N \tfrac{1}{c^n} \, \bigl ( f \magWnr g \bigr )_{(n)} + \tfrac{1}{c^{N+1}} \, R'_N(f,g) 
		\end{align}
		where all terms are known explicitly (cf.~equations (2.12) and (2.13) in \cite{Lein:two_parameter_asymptotics:2008}). Furthermore, each term of the expansion is in symbol class 
		\begin{align*}
			\bigl ( f \magWsr g \bigr )_{(n)} \in \Hoerr{m_1 + m_2 - 2 n \rho}
		\end{align*}
		and the remainder $R'_N(f,g)$ is an element of $\Hoerr{m_1 + m_2 - 2 (N+1) \rho}$ whose seminorms can be estimated uniformly in $\nicefrac{1}{c} \in [0,\nicefrac{1}{c_0}]$ for some $c_0 > 0$. 
\end{thm}
%
% subsection Non-relativistic pseudodifferential operators (end)

\subsection{Construction of $\pi$ and $u$ to finite order} % (fold)
\label{non_rel:existence_pi_u}
As mentioned before, there is no Moyal projection and Moyal unitary which satisfy the analog of equations~\eqref{semirel:eqn:4_defect_equations_symbols}. However, we can construct 
\begin{align*}
	\pi^{(k)} = \sum_{n = 0}^k \tfrac{1}{c^n} \, \pi_n \in \Hoermr{k}{1} 
	, 
	&&
	\pi_0 = \left (
	\begin{matrix}
		\id_{\C^2} & 0 \\
		0 & 0 \\
	\end{matrix}
	\right ) 
	, 
	\; 
	\pi_n \in \Hoermr{n}{1} 
	, 
\end{align*}
and 
\begin{align*}
	u^{(k)} = \sum_{n = 0}^k \tfrac{1}{c^n} \, u_n \in \Hoermr{k}{1} 
	, 
	&&
	u_0 = \id_{\C^4}
	, 
	\; 
	u_n \in \Hoermr{n}{1} 
	, 
\end{align*}
which satisfy these equations up to $\order(\nicefrac{1}{c^{k+1}})$. Since higher-order terms have stronger and stronger growth at infinity, these expansions are not asymptotic nor do their quantizations define bounded operators if $k \geq 1$. Nevertheless, we can regularize $\pi^{(k)}$ and $u^{(k)}$ using an energy cutoff $\chi_B$ (cf.~equation~\eqref{non_rel:eqn:regularizing_symbols_cutoff}), and the quantizations of these \emph{regularized} symbols are bounded and serve the roles of $\pi$ and $u$ in Theorem~\ref{semirel:thm:semi-relativistic_limit}. 
\begin{prop}\label{non_rel:prop:existence_pi}
	Let Assumption~\ref{intro:assumption:fields} on $V$ and $\mathbf{B}$ be satisfied. Then for any $k \in \N_0$ there exists a symbol 
	\begin{align*}
		\pi^{(k)} &= \sum_{n = 0}^k \tfrac{1}{c^n} \, \pi_n 
		\in \Hoermr{k}{1}
	\end{align*}
	with $\pi_n \in \Hoermr{n}{1}$ which satisfies 
	\begin{align}
		\pi^{(k)} \magWnr \pi^{(k)} - \pi^{(k)} &= \order(\nicefrac{1}{c^{k+1}}) \in \Hoermr{2k}{1} 
		\label{non_rel:eqn:projection_defect}
		\\
		\bigl [ \Hnr , \pi^{(k)} \bigr ]_{\magWnr} &= \order(\nicefrac{1}{c^{k+1}}) \in \Hoermr{k+1}{1} 
		. 
		\label{non_rel:eqn:commutation_defect}
	\end{align}
	The first four terms are given by 
	\begin{align}
		\pi^{(3)} &= \pi_0 + \frac{1}{c} \frac{1}{2m} \, \xi \cdot \alpha 
		- \frac{1}{c^2} \frac{1}{4 m^2} \, \xi^2 \, \beta + 
		\notag \\
		&\qquad \; \, +
		\frac{1}{c^3} \biggl ( \frac{\eps}{4 m^2} \, (\mathbf{B} \cdot \Sigma) \, \beta - \frac{1}{4 m^3} \, \xi^2 \,  (\xi \cdot \alpha) + \frac{\eps \, i}{2 m^2} \, (\nabla_x V \cdot \alpha) \beta \biggr )
		% + \order(\nicefrac{1}{c^4})
		\label{non_rel:eqn:pi_3_explicit}
	\end{align}
	where $\Sigma = (\Sigma_1,\Sigma_2,\Sigma_3)$ is the vector of $4 \times 4$ spin matrices 
	\begin{align*}
		\Sigma_j = \left (
		\begin{matrix}
			\sigma_j & 0 \\
			0 & \sigma_j \\
		\end{matrix}
		\right ) 
		, 
		&&
		j = 1 , 2 , 3 
		. 
	\end{align*}
\end{prop}
\begin{proof}
	Formally, we can use the recursion relations (3.2), (3.4), (3.5) and (3.8) in \cite{Teufel:adiabatic_perturbation_theory:2003} to construct $\pi^{(k)}$ order-by-order as a symbol. Assume we have constructed the first $k$ orders of $\pi^{(k)} = \sum_{n = 0}^k \tfrac{1}{c^n} \, \pi_n$ where $\pi_n \in \Hoermr{n}{1}$, \ie $\pi^{(k)}$ satisfies equations~\eqref{non_rel:eqn:projection_defect} and \eqref{non_rel:eqn:commutation_defect} up to errors of order $\order(\nicefrac{1}{c^{k+1}})$. 
	
	It turns out that the diagonal part of $\pi_{k+1}$ is determined by the »projection defect«, defined as the $k+1$th order term of 
	\begin{align}
		\pi^{(k)} \magWnr \pi^{(k)} - \pi^{(k)} &=: \tfrac{1}{c^{k+1}} \, G_{k+1} + \order(\nicefrac{1}{c^{k+2}}) 
		\label{non_rel:eqn:projection_defect_comp}
		\\
		&= \frac{1}{c^{k+1}} \, \sum_{\substack{j + l + n = k+1 \\ j , l \leq 3}} \bigl ( \pi_j \magWnr \pi_l \bigr )_{(n)} + \order(\nicefrac{1}{c^{k+2}})
		\notag 
	\end{align}
	is the sum of symbols of order $\Hoermr{j + l - 2n}{1} \subset \Hoermr{k+1}{1}$ by Theorem~\ref{non_rel:thm:properties_magWnr}. The left-hand side is in symbol class $\Hoermr{2k}{1}$ since $\pi^{(k)} \magWnr \pi^{(k)} \in \Hoermr{2k}{1}$ by Theorem~\ref{non_rel:thm:properties_magWnr}. Then $G_{k+1} \in \Hoermr{k+1}{1}$ holds and also the diagonal part 
	\begin{align}
		\pi_{k+1}^{\mathrm{d}} &= - \pi_0 \magWnr G_{k+1} \magWnr \pi_0 + (\id_{\C^4} - \pi_0) \magWnr G_{k+1} \magWnr (\id_{\C^4} - \pi_0)
		\notag \\
		&= - \pi_0 \, G_{k+1} \, \pi_0 + (\id_{\C^4} - \pi_0) \, G_{k+1} \, (\id_{\C^4} - \pi_0)
		\in \Hoermr{k+1}{1}
		\label{non_rel:eqn:proj_diag_comp}
	\end{align}
	is in the same symbol class. 
	
	The »commutation defect« 
	\begin{align}
		\Bigl [ \, \Hnr \; , \; &\pi^{(k)} + \tfrac{1}{c^{k+1}} \pi_{k+1}^{\mathrm{d}} \Bigr ]_{\magWnr} =: \tfrac{1}{c^{k+1}} F_{k+1} + \order(\nicefrac{1}{c^{k+2}}) 
		\label{non_rel:eqn:commutation_defect_comp}
		\\
		&= \frac{1}{c^{k+1}} \, \sum_{\substack{j + l + n = k+1 \\ j \leq 2 , \, l \leq k}} \Bigl ( 
		\bigl ( H_j \magWnr \pi_l \bigr )_{(n)} - \bigl ( \pi_l \magWnr H_j \bigr )_{(n)} \Bigr ) 
	 	\, + \notag \\
		&\qquad \biggl . 
		+ \frac{1}{c^{k+1}} \, \Bigl (
		\bigl ( H_0 \magWnr \pi_{k+1}^{\mathrm{d}} \bigr )_{(0)} - \bigl ( \pi_{k+1}^{\mathrm{d}} \magWnr H_0 \bigr )_{(0)} \Bigr ) \biggr ) + \order(\nicefrac{1}{c^{k+2}})
		\notag \\
		&= \frac{1}{c^{k+1}} \, \sum_{\substack{j + l + n = k+1 \\ j \leq 2 , \, l \leq k}} \Bigl ( 
		\bigl ( H_j \magWnr \pi_l \bigr )_{(n)} - \bigl ( \pi_l \magWnr H_j \bigr )_{(n)} \Bigr ) 
		+ \order(\nicefrac{1}{c^{k+2}})
		\notag 
	\end{align}
	can be used to determine the off-diagonal part. The term 
	\begin{align*}
		\bigl [ H_0 , \pi_{k+1}^{\mathrm{d}} \bigr ] = 0 
	\end{align*}
	cancels as $H_0$ commutes with $\pi_0$. Theorem~\ref{non_rel:thm:properties_magWnr} ensures that the $\order(\nicefrac{1}{c^{k+2}})$ term is a symbol in $\Hoermr{k+2}{1}$ and that $F_{k+1} \in \Hoermr{k+1}{1}$. Then the off-diagonal part as given by \cite[equation~(3.8)]{Teufel:adiabatic_perturbation_theory:2003}, 
	\begin{align}
		\pi_{k+1}^{\mathrm{od}} &= \pi_0 \magWnr F_{k+1} \magWnr (H_0 - m)^{-1} \magWnr (\id_{\C^4} - \pi_0) 
		+ \\
		&\qquad \qquad \qquad \qquad 
		- (\id_{\C^4} - \pi_0) \magWnr (H_0 - m)^{-1} \magWnr F_{k+1} \magWnr \pi_0 
		\\
		&= \pi_0 \, F_{k+1} \, (H_0 - m)^{-1} \, (\id_{\C^4} - \pi_0) - (\id_{\C^4} - \pi_0) \, (H_0 - m)^{-1} \, F_{k+1} \, \pi_0
		, 
		\label{non_rel:eqn:proj_offdiag_comp}
	\end{align}
	is also in the correct symbol class. Thus the $k+1$ term $\pi_{k+1} = \pi_{k+1}^{\mathrm{d}} + \pi_{k+1}^{\mathrm{od}}$ is an element of $\Hoermr{k+1}{1}$. 
	
	By the construction in \cite{Teufel:adiabatic_perturbation_theory:2003}, $\pi^{(k+1)} := \pi^{(k)} + \tfrac{1}{c^{k+1}} \, \pi_{k+1}$ satisfies equations~\eqref{non_rel:eqn:projection_defect} and \eqref{non_rel:eqn:commutation_defect} up to $\order(\nicefrac{1}{c^{k+2}})$, although the remainders are not small in the sense of symbol classes. 
	
	It remains to show that 
	\begin{align*}
		\pi_0 = \left (
		\begin{matrix}
			\id_{\C^2} & 0 \\
			0 & 0 \\
		\end{matrix}
		\right )
	\end{align*}
	satisfies the induction hypothesis: indeed, the fact that $(x,\xi) \mapsto \pi_0$ is the constant function allows us to compute 
	\begin{align*}
		\pi_0 \magWnr \pi_0 - \pi_0 &= \pi_0^2 - \pi_0 = 0 \in \Hoermr{1}{1}
	\end{align*}
	and 
	\begin{align*}
		\bigl [ \Hsr , \pi_0 \bigr ]_{\magWnr} &= \bigl [ \Hsr , \pi_0 \bigr ] 
		= \tfrac{1}{c} \bigl [ H_1 , \pi_0 \bigr ] \in \Hoermr{1}{1}
	\end{align*}
	directly. Hence, we can proceed and construct $\pi^{(k)} \in \Hoermr{k}{1}$ to any order. The calculation of the first three terms has been moved to Appendix~\ref{appendix:non_rel:projection}. 
\end{proof}
Similarly, we can compute the symbol that takes the place of the Moyal unitary. 
\begin{prop}\label{non_rel:prop:existence_u}
	Let Assumption~\ref{intro:assumption:fields} on $V$ and $\mathbf{B}$ be satisfied. Then for any $k \in \N_0$ there exists a symbol 
	\begin{align}
		u^{(k)} &= \sum_{n = 0}^k \tfrac{1}{c^n} \, u_n 
		\in \Hoermr{k}{1}
		\notag \\
		&= u_0 - \frac{1}{c} \frac{1}{2m} \, (\xi \cdot \alpha) \beta 
		- \frac{1}{c^2} \frac{1}{8 m^2} \, \xi^2 \, \id_{\C^4} + 
		\notag \\
		&\qquad \; \, +
		\frac{1}{c^3} \biggl ( \frac{\eps}{4 m^2} \, (\mathbf{B} \cdot \Sigma) - \frac{3}{16 m^3} \, \xi^2 \, (\xi \cdot \alpha) \beta - \frac{\eps \, i}{2 m^2} \, (\nabla_x V \cdot \alpha) \biggr )
		+ \order(\nicefrac{1}{c^4})
		\label{non_rel:eqn:u_3_explicit}
		% &= u_0 - \tfrac{1}{c} \tfrac{1}{2m} (\xi \cdot \alpha) \beta 
		% - \tfrac{1}{c^2} \tfrac{1}{8 m^2} \xi^2 \, \id_{\C^4} + 
		% \\
		% &\qquad +
		% \tfrac{1}{c^3} \Bigl ( \tfrac{\eps}{4 m^2} (B \cdot \Sigma) - \tfrac{3}{16 m^3} \xi^2 \,  (\xi \cdot \alpha) \beta - \tfrac{\eps \, i}{2 m^2} (\nabla_x V \cdot \alpha) \Bigr )
		% + \order(\nicefrac{1}{c^4})
	\end{align}
	with $u_n \in \Hoermr{n}{1}$ which satisfies 
	\begin{align}
		u^{(k)} \magWnr {u^{(k)}}^* - 1 &= \order(\nicefrac{1}{c^{k+1}}) \in \Hoermr{2k}{1} 
		, 
		\label{non_rel:eqn:unitarity_defect}
		\\
		{u^{(k)}}^* \magWnr u^{(k)} - 1 &= \order(\nicefrac{1}{c^{k+1}}) \in \Hoermr{2k}{1} , 
		\; 
		\notag \\
		u^{(k)} \magWnr \pi^{(k)} \magWnr {u^{(k)}}^* - \piref &= \order(\nicefrac{1}{c^{k+1}}) \in \Hoermr{3k}{1} 
		. 
		\label{non_rel:eqn:intertwining_defect}
	\end{align}
\end{prop}
\begin{proof}
	Similarly to the previous proof, we can modify the proof of \cite[Lemma~3.15]{Teufel:adiabatic_perturbation_theory:2003}: assume we have found $u^{(k)} \in \Hoermr{k}{1}$ which satisfies equations~\eqref{non_rel:eqn:unitarity_defect} and \eqref{non_rel:eqn:intertwining_defect} up to $\order(\nicefrac{1}{c^{k+1}})$. Then the »unitarity defect«
	\begin{align}
		u^{(k)} \magWnr {u^{(k)}}^* - 1 &=: \tfrac{1}{c^{k+1}} \, A_{k+1} + \order(\nicefrac{1}{c^{k+2}}) 
		\label{non_rel:eqn:unitarity_defect_comp}
		\\
		&= \tfrac{1}{c^{k+1}} \, \sum_{\substack{j + l + n = k+1 \\ j , l \leq k}} \bigl ( u_j \magWnr u_l^* \bigr )_{(n)} + \order(\nicefrac{1}{c^{k+2}})
		\notag
	\end{align}
	splits into an $\order(\nicefrac{1}{c^{k+2}})$ error term which is in symbol class $\Hoermr{2k}{1}$ by Theorem~\ref{non_rel:thm:properties_magWnr} and a $\tfrac{1}{c^{k+1}}$ contribution $A_{k+1} \in \Hoermr{k+1}{1}$. 
	
	The »intertwining defect« 
	\begin{align}
		\bigl ( u^{(k)} + \tfrac{1}{c^{k+1}} \, \tfrac{1}{2} A_{k+1} \bigr ) \magWnr \pi^{(k)} \magWnr \bigl ( u^{(k)} + \tfrac{1}{c^{k+1}} \, &\tfrac{1}{2} A_{k+1} \bigr )^* - \piref 
		=: 
		\notag \\
		&=: \tfrac{1}{c^{k+1}} \, B_{k+1} + \order(\nicefrac{1}{c^{k+2}})
		\label{non_rel:eqn:intertwining_defect_comp}
	\end{align}
	can also be written as the sum of a $\tfrac{1}{c^{k+1}}$ term $B_{k+1} \in \Hoermr{k+1}{1}$ and a remainder of class $\Hoermr{3k}{1}$. 
	
	Then setting 
	\begin{align}
		u_{k+1} := - \tfrac{1}{2} A_{k+1} + \bigl [ \piref , B_{k+1} \bigr ] 
		\label{non_rel:eqn:u_k_1_comp}
	\end{align}
	yields the next order term and $u^{(k+1)} := u^{(k)} + \tfrac{1}{c^{k+1}} \, u_{k+1}$ satisfies equations~\eqref{non_rel:eqn:unitarity_defect} and \eqref{non_rel:eqn:intertwining_defect} up to $\order(\nicefrac{1}{c^{k+2}})$. It remains to show $u_0 = \id_{\C^4}$ satisfies the induction assumption; that again follows from direct computation. 
	
	Lastly, the details of the computations for the first three terms have been moved to Appendix~\ref{appendix:non_rel:unitary}. 
\end{proof}
%
% subsection Existence of Moyal resolvent and energy regularization (end)

\subsection{Energy regularization} % (fold)
\label{non_rel:energy_regularization}
For the sake of our discussion, a cutoff function $\chi \in \Cont^{\infty}_{\mathrm{c}}(\R,[0,1])$ for a compact energy region $\Lambda \subset \R$ will be a smooth function with compact support which is identical to $1$ on $\Lambda$, $\chi \vert_{\Lambda} = 1$. Then if $1_{\Lambda}$ is the characteristic function for $\Lambda$, we have 
\begin{align}
	\chi(\Hd) \, 1_{\Lambda}(\Hd) = 1_{\Lambda}(\Hd) 
	\label{non_rel:eqn:smoothened_cutoff_cutoff_equal_cutoff}
\end{align}
by functional calculus. The distribution $\chi_B$ associated to the operator $\chi(\Hd) = \Opnr^A(\chi_B)$ is a smoothing symbol, \ie an element of $\Hoermr{-\infty}{1}$. To show that, we use results by Iftimie, Măntoiu and Purice \cite{Iftimie_Mantoiu_Purice:commutator_criteria:2008}: they derive Beals and Bony-type commutator criteria for \emph{scalar}-valued magnetic $\Psi$DOs. The extension to matrix-valued functions (Proposition~\ref{appendix:mag_PsiDOs:prop:inversion}) is rather simple since $\mathcal{B}(\C^4)$ is finite-dimensional and thus there is no question how to complete the algebraic tensor product of $\Cont^{\infty}_{\mathrm{b}}(T^* \R^3)$ and $\mathcal{B}(\C^4)$. Hence, the Beals- and Bony-type commutator criteria extend to matrix-valued functions and we get Proposition~\ref{appendix:mag_PsiDOs:prop:inversion}. The proof is a straightforward modification of the arguments in the proofs of Propositions~6.3 and 6.7 in \cite{Iftimie_Mantoiu_Purice:commutator_criteria:2008} and \cite[Proposition~8.7]{Dimassi_Sjoestrand:spectral_asymptotics:1999}, and the details can be found in Appendix~\ref{appendix:mag_PsiDOs:extension}. If the symbols are operator-valued and these operators act on an infinite-dimensional Hilbert space $\mathcal{H}$, then $\mathcal{B}(\mathcal{H})$ is not even separable. 

Note that $\pi^{(k)}$ has an important property: its quantization commutes with functions of $\Hd$ up to $\order(\nicefrac{1}{c^{k+1}})$ and thus we obtain 
\begin{lem}\label{non_rel:lem:commutator_chi_B_pi}
	Let $\pi^{(k)}$ be the symbol constructed in Proposition~\ref{non_rel:prop:existence_pi} and $\chi_B \in \Hoermr{-\infty}{1}$ be the symbol associated to a cutoff function $\chi \in \Cont^{\infty}_{\mathrm{c}}(\R,[0,1])$ from Proposition~\ref{appendix:mag_PsiDOs:prop:inversion}. Then 
	\begin{align}
		\bigl [ \chi_B , \pi^{(k)} \bigr ]_{\magWnr} = \order(\nicefrac{1}{c^{k+1}}) \in \Hoermr{-\infty}{1}
		\label{non_rel:eqn:commutator_chi_B_pi}
	\end{align}
	holds. 
\end{lem}
\begin{proof}
	A priori, we know from Theorem~\ref{non_rel:thm:properties_magWnr} that $\bigl [ \chi_B , \pi^{(k)} \bigr ]_{\magWnr} \in \Hoermr{-\infty}{1}$ is a smoothing symbol. To see the right-hand side is $\order(\nicefrac{1}{c^{k+1}})$, we combine $\bigl [ \Hnr , \pi^{(k)} \bigr ]_{\magWnr} = \order(\nicefrac{1}{c^{k+1}})$, the fact that $\chi_B$ can be written in terms of the Moyal resolvent $\bigl ( \Hnr - z \bigr )^{(-1)_{\mathrm{nr}}^B}$ via the Helffer-Sjöstrand formula, and 
	\begin{align*}
		\Bigl [ \bigl ( \Hnr - &z \bigr )^{(-1)_{\mathrm{nr}}^B} \; , \; \pi^{(k)} \Bigr ]_{\magWnr} 
		= \\
		&= - \bigl ( \Hnr - z \bigr )^{(-1)_{\mathrm{nr}}^B} \magWnr \Bigl [ \bigl ( \Hnr - z \bigr ) , \pi^{(k)} \Bigr ]_{\magWnr} \magWnr \bigl ( \Hnr - z \bigr )^{(-1)_{\mathrm{nr}}^B} 
		% \\
		% &
		= \order(\nicefrac{1}{c^{k+1}}) 
		. 
	\end{align*}
\end{proof}
The symbol $\chi_B \in \Hoermr{-\infty}{1}$ can be used to regularize any $f \in \Hoermr{m}{1}$, $m \in \R$, and we will systematically use the notation 
\begin{align}
	f_{\chi} := f \magWnr \chi_B \in \Hoermr{-\infty}{1} 
	. 
	\label{non_rel:eqn:regularizing_symbols_cutoff}
\end{align}
%
% Using equation~\eqref{non_rel:eqn:smoothened_cutoff_cutoff_equal_cutoff}
% 
% subsection Energy regularization (end)

\subsection{Approximate quantum dynamics} % (fold)
\label{non_rel:Approximate quantum dynamics}
Now we can proceed as in Section~\ref{semirel:quantum_dynamics}. However, since the non-relativistic limit is more singular than the semi-relativistic limit, we need to focus on states whose energy with respect to $\Hd$ is finite. This can be seen from the terms of the effective hamiltonian (cf.~equation~\eqref{non_rel:eqn:heff_explicit} below): $\Opnr^A(\heff)$ is a fourth-order operator while $\Hd$ is first-order. What is more, the most singular terms of $\Opnr^A(\heff)$ is proportional to $\tfrac{1}{c^4} \, \Opnr^A \bigl ( \abs{\xi}^4 \bigr )$, \ie of fourth order in $\nicefrac{1}{c}$ and should be a »perturbation« of the lower-order terms. Similarly, the symbols $\pi^{(k)}$ and $u^{(k)}$ are not asymptotic expansions associated to some almost projection and almost unitary on $L^2(\R^3,\C^4)$. In fact, for $k > 0$, their quantizations are unbounded and thus cannot be a projection or unitary in the operator sense. Hence, we need to regularize: so let $\chi_B \in \Hoermr{-\infty}{1}$ be as in Proposition~\ref{appendix:mag_PsiDOs:prop:inversion}. Then $\pi_{\chi} = \pi \magWnr \chi_B \in \Hoermr{-\infty}{1}$ and $u_{\chi} = u \magWnr \chi_B \in \Hoermr{-\infty}{1}$ are smoothing symbols and thus the associated non-relativistic $\Psi$DOs define bounded operators on $L^2(\R^3,\C^4)$. 

With the energy cutoff in place, we can derive approximate dynamics up to any order. However, for the sake of simplicity, we content ourselves with $k = 4$ and we will suppress ${\;}^{(4)}$ in the notation. For instance, we set 
\begin{align*}
	\pi &:= \pi^{(4)} \in \Hoermr{4}{1}
	,
	\\
	u &:= u^{(4)} \in \Hoermr{4}{1} 
	. 
\end{align*}
Then our main result reads: 
\begin{thm}[Non-relativistic limit]\label{non_rel:thm:non_rel_limit}
	Let Assumption~\ref{intro:assumption:fields} on $\mathbf{B}$, $A$ and $V$ be satisfied. Then for any compact $\Lambda \subset \R$ and a cutoff function $\chi \in \Cont^{\infty}_{\mathrm{c}}(\R,[0,1])$ with $\chi \vert_{\Lambda} = 1$, the dynamics associated to the non-relativistic quantization of 
	\begin{align}
		\heff :&= \piref \, \bigl ( u \magWnr \Hnr \magWnr u^* \bigr )^{(4)} \, \piref 
		\notag \\
		&= m \, \id_{\C^2} \oplus 0_{\C^2}
		+ \frac{1}{c^2} \biggl ( \frac{1}{2m} \xi^2 + V \biggr ) \oplus 0_{\C^2}
		- \frac{1}{c^3} \frac{\eps}{2 m} \mathbf{B} \cdot \sigma \oplus 0_{\C^2}
		+ \notag \\
		&\qquad 
		+ \frac{1}{c^4} \biggl ( 
			- \frac{1}{8 m^3} \, \sabs{\xi}^4 
			+ \frac{\eps}{4 m^2} \, \bigl ( \nabla_x V \wedge \xi \bigr ) \cdot \sigma 
			+ \frac{\eps^2}{8 m^2} \, \Delta V \biggr ) \oplus 0_{\C^2}
		. 
		\label{non_rel:eqn:heff_explicit}
	\end{align}
	approximate the full dynamics $e^{- i t \Hd}$ for initial states in $\ran 1_{\Lambda}(\Hd)$ in the following sense: 
	\begin{align}
		\norm{\Bigl ( e^{- i t \Hd} - \Opnr^A(u_{\chi}^*) \, e^{- i t \Opnr^A(\heff)} \, \Opnr^A(u) \Bigr ) \, \Opnr^A(\pi_\chi) \, 1_{\Lambda}(\Hd) } = \order \bigl ( \nicefrac{(1 + \abs{t})}{c^5} \bigr )
		\label{non_rel:eqn:non_rel_limit_dynamics}
	\end{align}
\end{thm}
\begin{remark}
	If we neglect some terms in $\heff$ and only expand $\heff^{(k)}$ to order $k < 4$, then we deduce from a simple Duhamel argument 
	\begin{align*}
		\norm{\Bigl ( e^{- i t \Opnr^A(\heff)} - e^{- i t \Opnr^A(\heff^{(k)})} \Bigr ) \, \Opnr^A(u_{\chi})} = \order(\nicefrac{1}{c^{k+1}}) 
		. 
	\end{align*}
	Combined with equation~\eqref{non_rel:eqn:non_rel_limit_dynamics}, this implies that the dynamics generated by $\Opnr^A(\heff^{(k)})$ approximate the full dynamics up to $\order(\nicefrac{1}{c^{k+1}})$. 
\end{remark}
\begin{proof}
	Let $\Lambda \subset \R$ be compact and $\chi \in \Cont^{\infty}_{\mathrm{c}}(\R,[0,1])$ be such that $\chi \vert_{\Lambda} = 1$. Then by Proposition~\ref{appendix:mag_PsiDOs:prop:inversion}, there exists $\chi_B \in \Hoermr{-\infty}{1}$ such that $\Opnr^A(\chi_B) = \chi(\Hd)$. Thus the quantizations of $\pi_{\chi} \in \Hoermr{-\infty}{1}$ and $u_{\chi} \in \Hoermr{-\infty}{1}$ define bounded operators on $L^2(\R^3,\C^4)$ by Theorem~\ref{non_rel:thm:properties_magWnr} and the magnetic Caldéron-Vaillancourt theorem \cite[Theorem~3.1]{Iftimie_Mantiou_Purice:magnetic_psido:2006}. 
	
	Computing the effective hamiltonian $\heff$ is somewhat involved and we give the details in Appendix~\ref{appendix:non_rel:heff}. Its quantization is self-adjoint on 
	\begin{align*}
		\mathcal{D} \bigl ( \Opnr^A(\heff) \bigr ) = H^4_A(\R^3,\C^2) \otimes L^2(\R^3,\C^2) 
	\end{align*}
	by Lemma~\ref{appendix:mag_PsiDOs:lem:selfadjointness}, keeping in mind that $\Opnr^A(\heff)$ acts trivially on the positronic subspace. Hence, all the objects in equation~\eqref{non_rel:eqn:non_rel_limit_dynamics} are bounded operators on $L^2(\R^3,\C^4)$ and it remains to show the right-hand side is not just finite but $\order(\nicefrac{1}{c^5})$. By equation~\eqref{non_rel:eqn:smoothened_cutoff_cutoff_equal_cutoff} and Proposition~\ref{appendix:mag_PsiDOs:prop:inversion}, the full dynamics on $\ran 1_{\Lambda}(\Hd)$ can be approximated by 
	\begin{align}
		e^{-i t \Hd} \, 1_{\Lambda}(\Hd) &= e^{- i t \Hd} \, \chi(\Hd)^2 \, 1_{\Lambda}(\Hd) 
		\notag \\
		&= e^{- i t \Hd} \, \Opnr^A \Bigl ( \chi_B \magWnr \bigl ( u^* \magWnr u - (u^* \magWnr u - 1) \bigr ) \magWnr \chi_B \Bigr ) \, 1_{\Lambda}(\Hd) 
		\notag \\
		&= e^{- i t \Hd} \, \Opnr^A(u_{\chi}^*) \, \Opnr^A(u_{\chi}) \, 1_{\Lambda}(\Hd) 
		+ \notag \\
		&\qquad \qquad 
		+ e^{- i t \Hd} \, \Opnr^A \bigl ( \chi_B \magWnr (u^* \magWnr u - 1 ) \magWnr \chi_B \bigr ) \, 1_{\Lambda}(\Hd) 
		\notag \\
		&= e^{- i t \Hd} \, \Opnr^A(u_{\chi}^*) \, \Opnr^A(u_{\chi}) \, 1_{\Lambda}(\Hd) + \order_{\norm{\cdot}}(\nicefrac{1}{c^5})
		. 
		\label{non_rel:eqn:Dirac_dynamics_u_chi}
	\end{align}
	The remainder 
	\begin{align*}
		e^{- i t \Hd} \, \Opnr^A \bigl ( \chi_B \magWnr (u^* \magWnr u - 1 ) \magWnr \chi_B \bigr ) \, 1_{\Lambda}(\Hd) \overset{\eqref{non_rel:eqn:unitarity_defect}}{=} \order_{\norm{\cdot}}(\nicefrac{1}{c^5})
	\end{align*}
	is the product of bounded operators as $\chi_B \magWnr ( u^* \magWnr u - 1 ) \magWnr \chi_B \in \Hoermr{-\infty}{1}$ and its quantization is bounded \cite[Theorem~3.1]{Iftimie_Mantiou_Purice:magnetic_psido:2006}. To compare full and effective dynamics, we use equation~\eqref{non_rel:eqn:Dirac_dynamics_u_chi} to make a Duhamel argument: 
	\begin{align*}
		\Bigl ( e^{- i t \Hd} - &\Opnr^A(u_{\chi}^*) \, e^{- i t \Opnr^A(\heff)} \, \Opnr^A(u_{\chi}) \Bigr ) \, \Opnr^A(\pi_{\chi}) \, 1_{\Lambda}(\Hd) = 
		\\
		&\overset{\eqref{non_rel:eqn:Dirac_dynamics_u_chi}}{=} \Bigl ( e^{- i t \Hd} \, \Opnr^A(u_{\chi}^*) \, \Opnr^A(u_{\chi}) - \Opnr^A(u_{\chi}^*) \, e^{- i t \Opnr^A(\heff)} \, \Opnr^A(u_{\chi}) \Bigr ) 
		\cdot \\
		&\qquad \qquad \cdot 
		% \, 
		\Opnr^A(\pi_{\chi}) \, 1_{\Lambda}(\Hd) + \order_{\norm{\cdot}}(\nicefrac{1}{c^5})
		\\
		&= \int_0^t \dd s \, \frac{\dd}{\dd s} \Bigl ( e^{- i s \Hd} \, \Opnr^A({u_{\chi}}^*) \, e^{- i (t-s) \Opnr^A(\heff)} \, \Opnr^A(u_{\chi}) \Bigr ) 
		\cdot \\
		&\qquad \qquad \cdot 
		% \, 
		\Opnr^A(\pi_{\chi}) \, 1_{\Lambda}(\Hd) + \order_{\norm{\cdot}}(\nicefrac{1}{c^5})
	\end{align*}
	The time derivative can be computed explicitly: since $\pi_{\chi} , u_{\chi} \in \Hoermr{-\infty}{1}$ and their quantizations are smoothing operators, and map any magnetic Sobolev space $H^m_A(\R^3,\C^4)$ onto $H^{\infty}_A(\R^3,\C^4)$ \cite[Proposition~3.14]{Iftimie_Mantiou_Purice:magnetic_psido:2006}. Thus, $H^{\infty}_A(\R^3,\C^4) \subseteq \mathcal{D}(\Hd) , \mathcal{D} \bigl ( \Opnr^A(\heff) \bigr )$ ensures  the time derivative exists in the strong operator topology sense and we compute 
	\begin{align}
		\frac{\dd}{\dd s} &\Bigl ( e^{- i s \Hd} \Opnr^A(u_{\chi}^*) \, e^{-i (t-s) \Opnr^A(\heff)} \, \Opnr^A(u_{\chi}) \Bigr ) \, \Opnr^A(\pi_{\chi}) 1_{\Lambda}(\Hd) = 
		\notag \\
		&= e^{- i s \Hd} \, \Opnr^A \Bigl ( \Hnr \magWnr \chi_B \magWnr u^* - \chi_B \magWnr u^* \magWnr \heff \Bigr ) 
		\cdot \notag \\
		&\qquad \qquad \qquad \; \cdot 
		e^{-i (t-s) \Opnr^A(\heff)} \, \Opnr^A \bigl ( u \magWnr \chi_B \magWnr \pi \bigr ) \, \chi(\Hd) \, 1_{\Lambda}(\Hd) 
		+ \order_{\norm{\cdot}}(\nicefrac{1}{c^5})
		\notag \\
		&= e^{- i s \Hd} \, \Opnr^A \Bigl ( \chi_B \magWnr \bigl ( \Hnr \magWnr u^* - u^* \magWnr \heff \bigr ) \Bigr ) 
		\cdot \notag \\
		&\qquad \qquad \qquad \; \cdot 
		% \, 
		e^{-i (t-s) \Opnr^A(\heff)} \, \Opnr^A \bigl ( u \magWnr \pi \magWnr \chi_B \bigr ) \, 1_{\Lambda}(\Hd) + \order_{\norm{\cdot}}(\nicefrac{1}{c^5})
		\notag \\
		&= e^{- i s \Hd} \, \Opnr^A \Bigl ( \chi_B \magWnr \bigl ( \Hnr \magWnr u^* - u^* \magWnr \heff \bigr ) \Bigr ) 
		\cdot \notag \\
		&\qquad \qquad \qquad \; \cdot 
		% \, 
		e^{-i (t-s) \Opnr^A(\heff)} \, \Opnr^A \bigl ( \piref \magWnr u \magWnr \chi_B \bigr ) \, 1_{\Lambda}(\Hd) + \order_{\norm{\cdot}}(\nicefrac{1}{c^5})
		\label{non_rel:eqn:derivative_Duhamel}
	\end{align}
	where we have used Lemma~\ref{non_rel:lem:commutator_chi_B_pi} and $u \magWnr \pi = \piref \magWnr u + \order(\nicefrac{1}{c^5}) \in \Hoermr{16}{1}$ in the second-to-last and last step, respectively. 
	
	The remainder terms in \eqref{non_rel:eqn:derivative_Duhamel} are all bounded operators due to the presence of $\chi_B$. Also, $\bigl [ \heff , \piref \bigr ]_{\magWnr} = 0$ implies $\bigl [ e^{- i t \Opnr^A(\heff)} , \Opnr^A(\piref) \bigr ] = 0$ and hence 
	\begin{align*}
		\eqref{non_rel:eqn:derivative_Duhamel} &= e^{- i s \Hd} \, \Opnr^A \Bigl ( \chi_B \magWnr \bigl ( \Hnr \magWnr u^* - u^* \magWnr \heff \bigr ) \magWnr \piref \Bigr ) 
		\cdot \notag \\
		&\qquad \qquad \qquad \; \cdot 
		% \, 
		e^{-i (t-s) \Opnr^A(\heff)} \, \Opnr^A \bigl ( u_{\chi} \bigr ) \, 1_{\Lambda}(\Hd) + \order_{\norm{\cdot}}(\nicefrac{1}{c^5})
		\\
		&= e^{- i s \Hd} \, \Opnr^A \Bigl ( \chi_B \magWnr \bigl ( \Hnr \magWnr u^* \magWnr \piref - u^* \magWnr \heff \bigr ) \Bigr ) 
		\cdot \notag \\
		&\qquad \qquad \qquad \; \cdot 
		% \, 
		e^{-i (t-s) \Opnr^A(\heff)} \, \Opnr^A \bigl ( u_{\chi} \bigr ) \, \chi(\Hd) + \order_{\norm{\cdot}}(\nicefrac{1}{c^5}) 
		. 
	\end{align*}
	A closer inspection of the difference in hamiltonians yields that it is $\order(\nicefrac{1}{c^5})$: from the defining relations of $\pi$ and $u$ (cf.~equations \eqref{non_rel:eqn:projection_defect}--\eqref{non_rel:eqn:commutation_defect} and \eqref{non_rel:eqn:unitarity_defect}--\eqref{non_rel:eqn:intertwining_defect}), we conclude 
	\begin{align*}
		\Hnr \magWnr u^* \magWnr \piref &= \Hnr \magWnr u^* \magWnr \piref \magWnr \piref 
		\\
		&
		= \Hnr \magWnr \pi \magWnr u^* \magWnr \piref + \order(\nicefrac{1}{c^5})
		\\
		&= u^* \magWnr u \magWnr \pi \magWnr \Hnr \magWnr u^* \magWnr \piref + \order(\nicefrac{1}{c^5})
		\\
		&= u^* \magWnr \piref \magWnr u \magWnr \Hnr \magWnr u^* \magWnr \piref + \order(\nicefrac{1}{c^5}) 
		\\
		&
		= u^* \magWnr \heff + \order(\nicefrac{1}{c^5}) 
		\in \Hoermr{17}{1}
		. 
	\end{align*}
	This means $\chi_B \magWnr \bigl ( \Hnr \magWnr u^* \magWnr \piref - u^* \magWnr \heff \bigr ) = \order(\nicefrac{1}{c^5}) \in \Hoermr{-\infty}{1}$ and its quantization defines a bounded operator on $L^2(\R^3,\C^4)$ by \cite[Theorem~3.1]{Iftimie_Mantiou_Purice:magnetic_psido:2006}. Hence, we have shown equation~\eqref{non_rel:eqn:non_rel_limit_dynamics}. 
\end{proof}
%
% subsection Approximate quantum dynamics (end)

\subsection{Spectral results} % (fold)
\label{non_rel:spectrum}
Analogously to Section~\ref{semirel:spectrum}, the »almost unitary equivalence« of $\Hd$ and $H_{\mathrm{eff}}^{(k)} = \Opnr^A \bigl ( \heff^{(k)} \bigr )$ implies that the spectra are related. 

With $\Pi^{(k)}_{\chi} := \Opnr^A \bigl ( \pi^{(k)}_{\chi} \bigr )$ and $U^{(k)}_{\chi} := \Opnr^A \bigl ( u^{(k)}_{\chi} \bigr )$ as shorthands for the regularized almost projection and unitary, we can formulate 
\begin{thm}\label{non_rel:thm:spectrum}
	Let $E_{\mathrm{max}} > 0$ be finite, and $\chi$ a cutoff function associated to $[0,E_{\mathrm{max}}]$ as in Proposition~\ref{appendix:mag_PsiDOs:prop:inversion}. Then for any $k \in \N_0$, the following statements hold true: 
	\begin{enumerate}[(i)]
		\item Let $E_0 \in \sigma(\Hd) \cap [0,E_{\mathrm{max}}]$. Then for any $\delta > 0$, there exists $\Psi_{\delta} \in L^2(\R^3,\C^4)$, $\snorm{\Psi_{\delta}} = 1$, such that 
		\begin{align*}
			\Bnorm{\bigl ( H_{\mathrm{eff}}^{(k)} - E_0 \bigr ) \, U^{(k)} \, \Pi^{(k)} \, \Psi_{\delta}} < C_k \, \delta + \order(\nicefrac{1}{c^{k+1}})
		\end{align*}
		holds where $C_k$ and the $\order(\nicefrac{1}{c^{k+1}})$ term are independent of $\delta$. 
		\item Similarly, if $E_0 \in \sigma \bigl ( H_{\mathrm{eff}}^{(k)} \bigr ) \cap [0,E_{\mathrm{max}}]$, then for any $\delta > 0$, there exists $\Psi_{\delta} \in L^2(\R^3,\C^4)$ with $\snorm{\Psi_{\delta}} = 1$ such that 
		\begin{align*}
			\Bnorm{\bigl ( \Hd - E_0 \bigr ) \, {U^{(k)}}^* \, \Pi^{(k)} \, \Psi_{\delta}} < C_k \, \delta + \order(\nicefrac{1}{c^{k+1}})
		\end{align*}
		holds where $C_k$ and the $\order(\nicefrac{1}{c^{k+1}})$ term are independent of $\delta$. 
	\end{enumerate}
\end{thm}
\begin{proof}
	Let $\chi_B$ be as in Proposition~\ref{appendix:mag_PsiDOs:prop:inversion} for $\Lambda = [0,E_{\mathrm{max}}]$ and $\chi \in \Cont^{\infty}_{\mathrm{c}}(\R)$. On states of finite energy, $\Pi^{(k)}_{\chi}$ and $U^{(k)}_{\chi}$ are an approximate unitary and approximate projection: in case of $\Pi^{(k)}_{\chi}$, the definitions of $\chi$ and $\chi_B$ imply
	\begin{align*}
		\Opnr^A(\chi_B) \; 1_{[0,E_{\mathrm{max}}]}(\Hd) = \chi(\Hd) \; 1_{[0,E_{\mathrm{max}}]}(\Hd) 
		= 1_{[0,E_{\mathrm{max}}]}(\Hd)
	\end{align*}
	and thus, up to $\order_{\norm{\cdot}}(\nicefrac{1}{c^{k+1}})$, $\Pi^{(k)}_{\chi}$ is a projection, 
	\begin{align*}
		{\Pi^{(k)}}^2 \, 1_{[0,E_{\mathrm{max}}]}(\Hd) &= \Opnr^A \bigl ( \pi^{(k)} \magWnr \chi_B \magWnr \pi^{(k)} \magWnr \chi_B \bigr ) \, 1_{[0,E_{\mathrm{max}}]}(\Hd) 
		\\
		&
		\overset{\eqref{non_rel:eqn:commutator_chi_B_pi}}{=} \Opnr^A \bigl ( \pi^{(k)} \magWnr \pi^{(k)} \magWnr \chi_B \magWnr \chi_B \bigr ) \, 1_{[0,E_{\mathrm{max}}]}(\Hd) + \order_{\norm{\cdot}}(\nicefrac{1}{c^{k+1}})
		\\
		% &= \Opnr^A \bigl ( \pi^{(k)} \magWnr \pi^{(k)} \magWnr \chi_B \bigr ) \, \chi(\Hd) \, 1_{[0,E_{\mathrm{max}}]}(\Hd) + \order_{\norm{\cdot}}(\nicefrac{1}{c^{k+1}})
		% \\
		&= \Opnr^A \bigl ( \pi^{(k)} \magWnr \chi_B \bigr ) \, \chi(\Hd) \, 1_{[0,E_{\mathrm{max}}]}(\Hd) + \order_{\norm{\cdot}}(\nicefrac{1}{c^{k+1}})
		\\
		&
		= \Pi^{(k)} \, 1_{[0,E_{\mathrm{max}}]}(\Hd) + \order_{\norm{\cdot}}(\nicefrac{1}{c^{k+1}}) 
		. 
	\end{align*}
	The arguments showing that $U^{(k)}_{\chi}$ is an almost-unitary on $\mathrm{ran} \, 1_{[0,E_{\mathrm{max}}]}$ are analogous. Then one can adapt the proof of Theorem~\ref{semirel:thm:spectrum} to obtain (i) and (ii). 
\end{proof}
%
% subsection Spectral considerations (end)
% section non_relativistic_limit (end)
%!TEX root = /Users/max/Dropbox/research/non- and semi-relativistic limit of the Dirac equation (FL 2008)/ahp_non_semi_rel_Dirac.tex
\section{Discussion and related literature} % (fold)
\label{discussion}
Roughly speaking, results which connect the Dirac operator to semi- or non-relativistic Pauli-type operators belong to one of three categories: approaches which relate the two dynamics, block diagonalization schemes and purely spectral results. Our work falls into the first category. 

These three categories are not independent, but form a hierarchy: most dynamical results include a systematic block diagonalization scheme. And since the block diagonalized operator is unitarily equivalent to $\widehat{H}_{\mathrm{D}}$, one knows that these two operators are isospectral, and studying the spectra of the operators in the block diagonals yields information on the spectrum of the Dirac operator.

\subsection{Approximation of dynamics} % (fold)
\label{discussion:our_summary}
To facilitate a comparison to other works, let us give a summary of our main results: using \emph{magnetic} pseudodifferential methods, we have shown that if the typical energies are small, \ie $\nicefrac{v}{c} \ll 1$, then Theorem~\ref{semirel:thm:semi-relativistic_limit} ensures there exists a projection $\Pi_{\mathrm{sr}}$, a unitary operator $U_{\mathrm{sr}}$ and an effective hamiltonian $H_{\mathrm{sr \, eff}} = \Opsr^A(h_{\mathrm{sr \, eff}})$ which approximate the dynamics to \emph{any} order in $\nicefrac{1}{c}$ in the sense that 
\begin{align*}
	\Bnorm{\Bigl ( e^{- i t \widehat{H}_{\mathrm{D}}} -  U_{\mathrm{sr}}^* \, e^{- i c^2 t H_{\mathrm{sr \, eff}}} \, U_{\mathrm{sr}} \Bigr ) \Pi_{\mathrm{sr}}} = \order \bigl ( \nicefrac{(1 + \abs{t})}{c^{\infty}} \bigr )
\end{align*}
holds. Note the extra factor $c^2$ which stems from rescaling the Dirac hamiltonian $\widehat{H}_{\mathrm{D}} = c^2 \, \Hd$ in Section~\ref{scalings}. The above equation also tells us that the block structure is preserved by the dynamics up to arbitrarily small error in $\nicefrac{1}{c}$. 

All operators involved, $\Pi_{\mathrm{sr}}$, $U_{\mathrm{sr}}$ and $H_{\mathrm{sr \, eff}}$, are $\order(\nicefrac{1}{c^{\infty}})$-close in norm to pseudodifferential operators which have an \emph{asymptotic} expansion in $\nicefrac{1}{c}$. Thereby, we have disproven a claim by Reiher and Wolf \cite{Reiher_Wolf:exact_decoupling_1:2004} that no unitary block diagonalizing $\widehat{H}_{\mathrm{D}}$ with expansion in $\nicefrac{1}{c}$ exists. Furthermore, we compute the first two and three orders of $\Pi_{\mathrm{sr}} = \Opsr^A(\pi_{\mathrm{sr} \, 0}) + \order_{\norm{\cdot}}(\nicefrac{1}{c^3})$, $U_{\mathrm{sr}} = \Opsr^A(u_{\mathrm{sr} \, 0}) + \order_{\norm{\cdot}}(\nicefrac{1}{c^3})$ and $H_{\mathrm{sr \, eff}}$, respectively, where the latter is the semi-relativistic magnetic pseudodifferential operator associated to 
\begin{align*}
	h_{\mathrm{sr \, eff}} &= E \, \piref + \tfrac{1}{c^2} V \, \piref 
	\, + \\
	&\quad 
	- \frac{1}{c^3} \frac{\eps}{2 E (E+m)} 
		% \Bigl ( 
			\bigl ( (E+m) \, \mathbf{B} - (\nabla_x V \wedge \xi) \bigr ) \cdot \sigma \oplus 0_{\C^2}
		% \Bigr ) 
	+ \order(\nicefrac{1}{c^4}) 
	. 
\end{align*}
Here, $E(\xi) = \sqrt{m^2 + \xi^2}$ is the relativistic kinetic energy. The third-order term is the first spin-dependent contribution and well-known from the description of relativistic spin dynamics via the T-BMT equation \cite{Bargmann_Michel_Telegdi:relativistic_spin_dynamics:1959,Spohn:semiclassics_Dirac:2000,PST:sapt:2002,Teufel:adiabatic_perturbation_theory:2003}. In principle, fourth- and higher-order terms of $H_{\mathrm{sr \, eff}}$ could be obtained with moderate computational effort. 

If the typical velocities and momenta are an order of magnitude smaller, then for electronic states of finite energy, \ie those from $\ran 1_{[0,E_0]} \bigl ( \tfrac{1}{c^2} \, \widehat{H}_{\mathrm{D}} \bigr )$, one can construct an approximate projection $\Pi_{\mathrm{nr}}^{(k)}$ and an approximate unitary $U_{\mathrm{nr}}^{(k)}$ on $\ran 1_{[0,E_0]} \bigl ( \tfrac{1}{c^2} \, \widehat{H}_{\mathrm{D}} \bigr )$ so that the full non-relativistic dynamics can be approximated up to errors of order $\order(\nicefrac{1}{c^{k+1}})$ by $e^{- i c^2 t H^{(k)}_{\mathrm{nr \, eff}}}$ where $H^{(k)}_{\mathrm{nr \, eff}}$ is the non-relativistic magnetic pseudodifferential operator associated to 
\begin{align*}
	h^{(k)}_{\mathrm{nr \, eff}} &= m \, \piref
	+ \frac{1}{c^2} \biggl ( \frac{1}{2m} \xi^2 + V \biggr ) \oplus 0_{\C^2}
	- \frac{1}{c^3} \frac{\eps}{2 m} \mathbf{B} \cdot \sigma \oplus 0_{\C^2}
	+ \notag \\
	&\quad 
	+ \frac{1}{c^4} \biggl ( 
		- \frac{1}{8 m^3} \, \sabs{\xi}^4 
		+ \frac{\eps}{4 m^2} \, \bigl ( \nabla_x V \wedge \xi \bigr ) \cdot \sigma 
		+ \frac{\eps^2}{8 m^2} \, \Delta V \biggr ) \oplus 0_{\C^2}
	+ \order(\nicefrac{1}{c^5})
	. 
\end{align*}
In other words, Theorem~\ref{non_rel:thm:non_rel_limit} tells us that 
\begin{align*}
	\norm{\Bigl ( e^{- i t \widehat{H}_{\mathrm{D}}} - {U_{\mathrm{nr}}^{(k)}}^* \, e^{- i c^2 t H^{(k)}_{\mathrm{nr \, eff}}} \, U_{\mathrm{nr}}^{(k)} \Bigr ) \, \Pi_{\mathrm{nr}}^{(k)} \, 1_{[0,E_0]} \bigl ( \tfrac{1}{c^2} \, \widehat{H}_{\mathrm{D}} \bigr ) } = \order \bigl ( \nicefrac{(1 + \abs{t})}{c^{k-1}} \bigr )
\end{align*}
holds for any integer $k \geq 2$. Even though it is in principle possible to compute fifth- and higher-order terms of $H^{(k)}_{\mathrm{nr \, eff}}$, the computational effort increases considerably. 

The fact that we obtain semi- and non-relativistic limit using \emph{precisely the same method} in different scalings substantiates the claim found in physics text books \cite{Bjorken_Drell:relativistic_qm:1998}, namely that if one were to compute the Foldy-Wouthuysen transform of $\Hd$, the terms are related to the Taylor expansion of $\sqrt{m^2 + \xi^2}$ for small momenta. We can make that claim much more precise: if we Taylor-expand $h_{\mathrm{sr \, eff}}(x,\nicefrac{\xi}{c})$ around $\nicefrac{1}{c} = 0$, we obtain 
\begin{align*}
	h_{\mathrm{sr \, eff}}(x,\nicefrac{\xi}{c}) &= m \, \piref + \frac{1}{c^2} \left ( \frac{1}{2m} \xi^2 + V \right ) \, \piref 
	- \frac{1}{c^3} \frac{\eps}{2 m} B \cdot \sigma \oplus 0_{\C^2} 
	+ \\
	&\quad 
	+ \frac{1}{c^4} \left ( - \frac{1}{8 m^3} \, \abs{\xi}^4 + \frac{\eps}{4 m^2} \, (\nabla_x V \wedge \xi) \cdot \sigma \right ) \oplus 0_{\C^2}
	% + \\
	% &\qquad 
	+ \order(\nicefrac{1}{c^5}) 
	\\
	&= h_{\mathrm{nr \, eff}}(x,\xi) + \order(\nicefrac{1}{c^4})
	. 
\end{align*}
Perhaps surprisingly, we recover two out of three terms of $h_{\mathrm{nr \, eff \, 4}}$ right away, the Darwin term is a »genuine« fourth-order term stemming from a semi-relativistic fourth-order correction. Certainly, this suggestive computation needs to be taken with a grain of salt: the expansion of $\sqrt{m^2 + \xi^2}$ has a finite radius of convergence and thus, the above only holds if $\abs{\xi}$ is small enough. Hence, the pathological nature of the non-relativistic approximation is present even in the classical system. It is nevertheless revealing and satisfactory to connect semi- and non-relativistic hamiltonians in such a simple way.  

The assumption that $B$ and $V$ are of class $\Cont^{\infty}_{\mathrm{b}}$ stem from our use of pseudodifferential methods. Admitting suitable matrix-valued potentials of class $\Cont^{\infty}_{\mathrm{b}}$ is straightforward. If one wants to generalize to less regular potentials, the $\Psi$DO techniques will need to be augmented by operator theoretic methods, and some arguments which follow from results on pseudodifferential operators need to be completed »by hand«. On the positive side, we make no decay assumptions on $B$, $A$ and $V$, and our assumptions include the case of constant magnetic field. 

The crucial ingredient in semi- and non-relativistic limit was choosing an \emph{adapted} \emph{magnetic} pseudodifferential calculus where each of the attributes is crucial: since powers of $\nicefrac{1}{c}$ enter as prefactors in front of the gradient \emph{and} the magnetic vector potential, attempting to derive our results using \emph{non}-magnetic pseudodifferential theory is doomed to fail. In addition, our results extend naturally to magnetic fields of class $\Cont^{\infty}_{\mathrm{b}}$ whereas the standard pseudodifferential approach to describe magnetic systems, standard Weyl quantization combined with minimal substitution, assume that the components of the \emph{vector potential} are $\Cont^{\infty}_{\mathrm{b}}$. This then excludes the physically important case of constant magnetic fields. 

Lastly, let us mention our spectral results, Theorems~\ref{semirel:thm:spectrum} and \ref{non_rel:thm:spectrum}. These are by no means particularly deep, but we felt it necessary to include them given the breadth of spectral results (cf.~Section~\ref{discussion:spectral} below). In essence, they are just saying that if $\Hd$ has spectrum in the vicinity of $E_0 > 0$, then $H_{\mathrm{eff}}$ has spectrum in a possibly larger neighborhood of $E_0$ and vice versa. 
\medskip

\noindent
From a physical perspective, approximating the \emph{dynamics} is \emph{the} crucial step in the justification of why one can use the semi- or non-relativistic Pauli equation to describe a quantum spin-$\nicefrac{1}{2}$ particle for small energies. These models are (conceptually and numerically) simpler than treating the full Dirac equation \eqref{intro:eqn:Dirac_equation}. 

Our results establish a hierarchy of approximations in the following sense: if the typical energies are small compared to the rest energy of the particle, then $\nicefrac{1}{c} \ll 1$, and electronic and positronic degrees decouple to any order in $\nicefrac{1}{c}$. The dynamics for electronic states are then generated by the semi-relativistic hamiltonian for positive energy initial states, $H_{\mathrm{sr \, eff}}$. For even smaller energies, the dynamics generated by the non-relativistic effective hamiltonian $H_{\mathrm{nr \, eff}}$ approximate $e^{- i t \widehat{H}_{\mathrm{D}}}$ for finite-energy states. 

Even though our work is not the first one to show how to approximate the dynamics of a Dirac particle, to the best of our knowledge, it is the first to show that the dynamics can be approximated by semi- or non-relativistic dynamics to \emph{any} order in $\nicefrac{1}{c}$. The works of Bechouche et al \cite{Bechouche_Mauser_Poupaud:non_relativistic_limit_dynamics_Dirac:1998} and Mauser \cite{Mauser:non_relativistic_limit_dynamics_Dirac:1999} are the first to derive a non-relativistic limit of the Dirac \emph{dynamics} using the same small parameter, the ratio of a characteristic velocity to the speed of light. Their main result in this context, \cite[Corollary~5.1]{Bechouche_Mauser_Poupaud:non_relativistic_limit_dynamics_Dirac:1998}, shows that the dynamics generated by the Pauli hamiltonian
\begin{align*}
	H_{\mathrm{P}} = \frac{1}{2 m} \bigl ( - i \eps \nabla_x - \tfrac{1}{c} A(\hat{x}) \bigr )^2 + V(\hat{x}) - \frac{1}{c} \, \frac{\eps}{2m} \mathbf{B}(\hat{x}) \cdot \sigma 
\end{align*}
approximate the full Dirac dynamics up to errors of $\order(\nicefrac{1}{c^2})$ for times of $\order(1)$ and initial states that are in some sense $\order(\nicefrac{1}{c^2})$-close to $\ran \pi_{\mathrm{sr} \, 0} \bigl (- i \tfrac{\eps}{c} \nabla_x \bigr )$. Bechouche et al's result is stronger than Theorem~\ref{non_rel:thm:non_rel_limit} in two ways: they include time-dependent fields and require less regularity from the components of $A$ and $V$ -- and thus from $\mathbf{E}$ and $\mathbf{B}$. On the other hand, they provide no systematic perturbation scheme, and it is not clear how to extend their ideas to include higher-order corrections to $H_{\mathrm{P}}$. Furthermore, Theorem~\ref{non_rel:thm:non_rel_limit} holds for long times and magnetic fields which do not decay at infinity: the difference between Dirac and non-relativistic dynamics goes to $0$ \emph{in norm} as $\nicefrac{1}{c} \rightarrow 0$ even for times of order $\order(c^{\alpha})$, $\alpha < 5$. 

Other previous results in this direction, \eg the works by Brummelhuis and Nourigat \cite{Brummelhuis_Nourrigat:scattering_Dirac:1999}, Spohn \cite{Spohn:semiclassics_Dirac:2000} and Panati, Spohn and Teufel \cite{PST:sapt:2002}, all use the semiclassical parameter $\eps$ as expansion parameter. Hence, these authors do not recover semi- or non-relativistic limiting dynamics, but effective dynamics for slowly varying external fields. Physically, this distinction is indeed significant: while the aforementioned three publications obtain the T-BMT equation for spin, a »semiclassical limit of $H_{\mathrm{sr \, eff}}$ in $\nicefrac{1}{c}$« (itself a straightforward consequence of \cite[Theorem~3.6.2]{Lein:progress_magWQ:2010}) yields ballistic motion, and we are unable to say anything about spin dynamics since the first two terms of $H_{\mathrm{sr \, eff}}$ are scalar in the spin degrees of freedom. 
% subsection Approximation of dynamics (end)

\subsection{Block diagonalization methods} % (fold)
\label{discussion:block_diagonalization}
Block-diagonalization techniques are recurrent schemes which successively construct a unitary operator $U_k$ such that 
\begin{align*}
	U_k \, H_k \, U_k^* = U_k \left (
	\begin{matrix}
		h_{+ \, k} & \order(\epsilon^{k+1}) \\
		\order(\epsilon^{k+1}) & h_{- \, k} \\
	\end{matrix}
	\right ) U_k^* = \left (
	\begin{matrix}
		h_{+ \, k+1} & \order(\epsilon^{k+2}) \\
		\order(\epsilon^{k+2}) & h_{- \, k+1} \\
	\end{matrix}
	\right ) =: H_{k+1}
\end{align*}
where $H_0 := \widehat{H}_{\mathrm{D}}$ is the original Dirac hamiltonian and $\epsilon$ some small parameter (cf.~Section~\ref{scalings:small_parameter}), \eg the mass $m$ or the charge $\mathrm{e}$. Then the Dirac hamiltonian and the approximately block diagonalized hamiltonian $H_{k+1}$ is related to the Dirac hamiltonian via 
\begin{align*}
	H_{k+1} = \bigl ( U_k \cdots U_1 \bigr ) \, \widehat{H}_{\mathrm{D}} \, \bigl ( U_k \cdots U_1 \bigr )^* 
	. 
\end{align*}
The first such scheme was proposed by Foldy and Wouthuysen in their famous paper \cite{Foldy_Wouthuysen:Dirac_theory:1950}; although many subsequent publications claim they expand in $\epsilon = \nicefrac{1}{c}$, they in fact set $c = 1$ and use $\nicefrac{1}{m}$ as small parameter. In the presence of external fields, they arrive at the non-relativistic hamiltonian $H_{\mathrm{nr \, eff}}^{(4)}$ which is thus plagued by technical problems of higher-order terms being increasingly singular. This not only poses problems for mathematicians, but also for physicists and theoretical chemists doing numerical calculations of relativistic systems \cite{Reiher_Wolf:exact_decoupling_1:2004}. 

Hence, the search for alternative block diagonalization schemes is the subject of active research. One such method is the so-called Douglas--Kroll--Heß method where the small parameter is essentially the charge $\mathrm{e}$. 
The idea for the DKH transformation  \cite{Reiher_Wolf:exact_decoupling_1:2004,Reiher:relativistic_DK_theory:2012} originated in a paper by Douglas and Kroll \cite{Douglas_Kroll:QED_corrections_helium:1974} whose ideas were expanded on by Heß and co-workers \cite{Hess:revision_DK_transform:1986,Hess_Jansen:revision_DK_transform:1989}. Meanwhile, Siedentop and Stockmeyer have analyzed the works of Reiher et al.~from a mathematical point of view \cite{Siedentop_Stockmeyer:DKH_method_convergence:2006}. For the Coulomb potential, they show two things: first of all, they prove that the spectra of the upper-left and lower-right block operators of $H_k$ converge to positive and negative part of the spectrum of $\widehat{H}_{\mathrm{D}}$ as $k \rightarrow \infty$ (Theorem~2). Secondly, they show that the  projection onto the electronic states $\Pi(\mathrm{e})$, the block diagonalizing unitary $U(\mathrm{e})$ and in some sense $H_k$ are not just asymptotic, but \emph{analytic} in $\mathrm{e}$ (Lemma~1, Theorem~1 and Lemma~6, respectively). To the best of our knowledge, there are no mathematical works on the DKH method which include magnetic fields. Even from a non-rigorous perspective, it seems less clear how to incorporate magnetic fields (cf.~the discussion on Magnetic Properties in \cite[p.~144]{Reiher:relativistic_DK_theory:2012}). 

Lastly, let us mention the works of Cordes \cite{Cordes:pseudodifferential_FW_transform:1983,Cordes:pseudodifferential_FW_transform:2004} who was the first to formulate the problem in the language of pseudodifferential theory: he uses no small parameter, but instead classifies the terms according to their symbol class (\ie decay in momentum). One glance at Definition~\ref{semirel:defn:semi-relativistic_symbol} reveals that ordering according to decay is not the same as ordering with respect to a small parameter. As he uses the standard expansion of the non-magnetic Moyal product, his result is a precursor to \cite{Brummelhuis_Nourrigat:scattering_Dirac:1999} and \cite[Theorem~4.4]{Teufel:adiabatic_perturbation_theory:2003}. 
% subsection Block diagonalization schemes (end)

\subsection{Spectral results} % (fold)
\label{discussion:spectral}
Most publications on the $\nicefrac{1}{c} \rightarrow 0$ limit focus on the spectral aspects, in particular on bound states (\eg \cite{Hunziker:non_rel_limit_Dirac:1975,Gesztesy_Grosse_Thaller:relativist_corrections_bound_state_energies:1984,Grigore_Nenciu_Purice:non_rel_limit_Dirac:1989}, see also \cite{Thaller:Dirac_equation:1992} and references therein). Grigore, Nenciu and Purice \cite{Grigore_Nenciu_Purice:non_rel_limit_Dirac:1989}, for instance, combine pseudoresolvents with analytic perturbation theory in the sense of Kato to prove that to lowest order, the electronic half of the spectrum $\sigma(\widehat{H}_{\mathrm{D}}) \cap [0,+\infty)$ becomes arbitrarily close to the spectrum of the Pauli operator as $\nicefrac{1}{c} \rightarrow 0$ (Theorem~I.4). They show that if $E_0 \in \sigma(\widehat{H}_{\mathrm{D}}) \cap [0,+\infty)$ is an isolated eigenvalue of finite degeneracy, then it can be computed in terms of the eigenfunction and the Pauli operator with higher-order corrections (Theorem~III.1). 

Note that most of the time (\eg \cite[Chapter~6]{Thaller:Dirac_equation:1992}, \cite{Hunziker:non_rel_limit_Dirac:1975,Gesztesy_Grosse_Thaller:relativist_corrections_bound_state_energies:1984,Grigore_Nenciu_Purice:non_rel_limit_Dirac:1989}), the magnetic field is scaled \emph{differently} than in the physics literature: $B$ is replaced with $c B$ and $A$ with $c A$, \ie the magnetic field goes up the spout as $\nicefrac{1}{c} \rightarrow 0$. In Thaller's words, the rationale behind this choice of scaling is to avoid \emph{»turning the light off«.} For otherwise magnetic terms were higher-order effects and the leading-order hamiltonian would be the ordinary \emph{non}-magnetic Schrödinger operator rather than the Pauli operator. 

Looking at the non-relativistic effective hamiltonian, equation~\eqref{non_rel:eqn:heff_explicit}, we see that this is not necessary, effects which stem from the presence of the magnetic field simply \emph{are} higher-order effects and appear at \emph{third} order in $\nicefrac{1}{c}$. 
% subsection Spectral results (end)

% section Discussion and literature (end)
%!TEX root = /Users/max/Dropbox/research/non- and semi-relativistic limit of the Dirac equation (FL 2008)/non_semi_rel_Dirac.tex
%
\begin{appendix}
	\section{Magnetic $\Psi$DOs} % (fold)
	\label{appendix:mag_PsiDOs}
	This section provides supplementary material concerning the magnetic pseudodifferential calculi of matrix-valued symbols.

	\subsection{Extension to matrix-valued symbols} % (fold)
	\label{appendix:mag_PsiDOs:extension}
	Even though extending results concerning $\Psi$DOs to matrix-valued functions is straightforward and standard, we discuss some central results for completeness and convenience. Let 
		\begin{align*}
			\Hoermr{m}{\rho}(T^* \R^d) := \Bigl \{ f \in \Cont^{\infty}(T^* \R^d) \; &\big \vert \; \forall a , \alpha \in \N_0^d \; \exists \, C_{a\alpha} > 0 : 
			\Bigr . \\
			&\qquad \qquad \Bigl . 
			\babs{\partial_x^a \partial_{\xi}^{\alpha} f(x,\xi)} \leq C_{a\alpha} \, \sqrt{1 + \xi^2}^{\, m - \abs{\alpha} \rho} \Bigr \}
		\end{align*}
		be the space of scalar-valued Hörmander symbols of order $m$ and type $\rho \in [0,1]$. The space of $\mathcal{B}(\C^N)$-valued Hörmander symbols $\Hoermr{m}{\rho} \bigl ( T^* \R^d , \mathcal{B}(\C^N) \bigr )$ is defined analogously to equation~\eqref{semirel:eqn:Hoermander_symbols}. Then 
		\begin{align}
			\Hoermr{m}{\rho} \bigl ( T^* \R^d , \mathcal{B}(\C^N) \bigr ) = \Hoermr{m}{\rho}(T^* \R^d) \otimes \mathcal{B}(\C^N)
			\label{appendix:non_rel:eqn:Hoermander_classes_tensor_product}
		\end{align}
		agree as Fréchet spaces by \cite[Proposition~40.2]{Treves:topological_vector_spaces:1967} and the finite-dimension\-ality of $\mathcal{B}(\C^N)$. First, we extend \cite[Theorem~4.1]{Iftimie_Mantiou_Purice:magnetic_psido:2006} since we need to know the domains of selfadjointness explicitly for the proofs of Theorems~\ref{semirel:thm:semi-relativistic_limit} and \ref{non_rel:thm:non_rel_limit}. 
	\begin{lem}\label{appendix:mag_PsiDOs:lem:selfadjointness}
		Assume $\mathbf{B} = \nabla_x \wedge A$ is of class $\Cont^{\infty}_{\mathrm{b}}$ and the components of $A$ are $\Cont^{\infty}_{\mathrm{pol}}$ functions. Let $H \in \Hoermr{m}{\rho} \bigl ( T^* \R^3 , \mathcal{B}(\C^N) \bigr )$, $\rho \in (0,1]$, be a symmetric $N \times N$ matrix-valued symbol such that 
		\begin{align*}
			(x,\xi) \mapsto \inf \abs{\sigma \bigl ( H(x,\xi) \bigr )} 
		\end{align*}
		is elliptic in the usual sense. Then $\Opnr^A(H)$ and $\Opsr^A(H)$ define selfadjoint operators on $H^m_A(\R^3) \otimes \C^N$ where $H^m_A(\R^3)$ is the $m$th magnetic Sobolev space (cf.~\cite[Definition~3.4]{Iftimie_Mantiou_Purice:magnetic_psido:2006}). 
	\end{lem}
	\begin{proof}
		The proof does not rely on the presence of a small parameter. The ellipticity of $(x,\xi) \mapsto \inf \abs{\sigma \bigl ( H(x,\xi) \bigr )}$ ensures there exists $R \geq 0$ such that the matrix $H(x,\xi)$ is invertible for all $x$ and $\abs{\xi} \geq R$. Hence, we can construct a parametrix for $H$ with respect to either $\magWsr$ or $\magWnr$ as in \cite[Theorem~2.4]{Iftimie_Mantiou_Purice:magnetic_psido:2006} and thus retrace the steps of the proof \cite[Theorem~4.1]{Iftimie_Mantiou_Purice:magnetic_psido:2006} with obvious modifications. 
	\end{proof}
	\begin{cor}\label{appendix:mag_PsiDOs:cor:selfadjointness_H_D}
		Under the assumptions of Lemma~\ref{appendix:mag_PsiDOs:lem:selfadjointness}, $\widehat{H}_{\mathrm{D}}$ and $\Hd$ define selfadjoint operators on $H^1_A(\R^3) \otimes \C^4$. 
	\end{cor}
	\begin{proof}
		As we can diagonalize the matrix $\Hsr(x,\xi)$ using $u_{\mathrm{sr} \, 0}(\xi)$, the function 
		\begin{align*}
			(x,\xi) \mapsto \inf \, \babs{\sigma \bigl ( \Hsr(x,\xi) \bigr )}
			= \inf \abs{\pm \sqrt{m^2 + \xi^2} + \tfrac{1}{c^2} V(x)}
		\end{align*}
		is obviously elliptic. Hence, Lemma~\ref{appendix:mag_PsiDOs:lem:selfadjointness} yields that $\Opsr^A(\Hsr) = \Hd$ is selfadjoint on $H^1_A(\R^3) \otimes \C^4$. 
	\end{proof}
	Generalizing commutator criteria is rather easy in our context, because the symbols take values in a \emph{finite-dimensional} Banach space $\mathcal{B}(\mathcal{H}) = \mathcal{B}(\C^N)$ (as opposed to non-separable if $\dim \mathcal{H} = \infty$). 
	Equation~\eqref{appendix:non_rel:eqn:Hoermander_classes_tensor_product} immediately implies the Beals criterion for matrix-valued symbols which identifies operators which are the magnetic quantization of a $\Cont^{\infty}_{\mathrm{b}}$ function. 
	\begin{cor}[Beals criterion]
		Using the notation of \cite{Iftimie_Mantoiu_Purice:commutator_criteria:2008}, we have 
		\begin{align*}
			\Cont^{\infty}_{\mathrm{b}} \bigl ( T^* \R^d , \mathcal{B}(\C^N) \bigr ) \cong \Cont^{\infty}_{\mathrm{b}}(T^* \R^d) \otimes \mathcal{B}(\C^N) \cong \Cont^{\infty} \bigl ( \mathfrak{T}^B , \mathfrak{C}^B \bigr ) \otimes \mathcal{B}(\C^N)
			. 
		\end{align*}
	\end{cor}
	\begin{proof}
		The first equivalence follows from noting $\Cont^{\infty}_{\mathrm{b}}(T^* \R^d) = \Hoermr{0}{0}(T^* \R^d)$ and  \eqref{appendix:non_rel:eqn:Hoermander_classes_tensor_product}. The second equivalence is the content of \cite[Theorem~2.5]{Iftimie_Mantoiu_Purice:commutator_criteria:2008}. 
	\end{proof}
	Similarly, one can extend the Beals- and Bony-type commutator criteria to deduce when a tempered distribution $F \in \mathcal{S}'(T^* \R^d,\mathcal{B}(\C^N))$ is really a Hörmander symbol of order $m$ and type $\rho$ (cf.~Theorems~5.2 and 5.5 in \cite{Iftimie_Mantoiu_Purice:commutator_criteria:2008}). Hence, the proof of the next Proposition follows from the arguments in \cite[Section~6]{Iftimie_Mantoiu_Purice:commutator_criteria:2008}.  
	\begin{prop}\label{appendix:mag_PsiDOs:prop:inversion}
		Let $H$, $\Op^A$, $\magW$ and ${\;}^{(-1)^B}$ denote either semi- or non-relativistic hamiltonian, quantization, Moyal product and Moyal inverse, respectively, \ie $H = \Hsr , \Hnr \in \Hoermr{1}{1}$, $\Op^A = \Opsr^A , \Opnr^A$, $\magW = \magWsr , \magWnr$ and ${\;}^{(-1)^B} = {\;}^{(-1)^B_{\mathrm{sr}}} , {\;}^{(-1)^B_{\mathrm{nr}}}$. Then the following holds: 
		\begin{enumerate}[(i)]
			\item For $z \in \C \setminus \R$, the Moyal resolvent $\bigl ( H - z \bigr )^{(-1)^B}$, \ie the distribution which satisfies 
		\begin{align*}
			\bigl ( H - z \bigr ) \magW \bigl ( H - z \bigr )^{(-1)^B} = 1 
			= \bigl ( H - z \bigr )^{(-1)^B} \magW \bigl ( H - z \bigr )
			, 
		\end{align*}
		exists as a Hörmander symbol of class $\Hoermr{-1}{1}$. 
			\item For any smooth cutoff function $\chi \in \Cont^{\infty}_{\mathrm{c}}(\R,[0,1])$ associated to a compact set $\Lambda \subset \R$, there exists a symbol $\chi_B \in \Hoermr{-\infty}{1}$ such that $\Op^A(\chi_B) = \chi(\Hd)$. 
		\end{enumerate}
	\end{prop}
	\begin{proof}
		The proofs do not rely on the presence of small parameters and thus hold for both, semi- and non-relativistic pseudodifferential operators. 
		\begin{enumerate}[(i)]
			\item Let $z \in \C \setminus \R$. Due to Assumption~\ref{intro:assumption:fields} on $\mathbf{B}$, $V$ and $A$ and Corollary~\ref{appendix:mag_PsiDOs:cor:selfadjointness_H_D}, we know that $\Op^A(H) = \Hd = {\Hd}^*$ defines a selfadjoint operator on $H^1_A(\R^3,\C^4)$, the first magnetic Sobolev space, and hence $z \not \in \sigma(\Hd) \subseteq \R$. Hence, the resolvent $\bigl ( \Hd - z \bigr )^{-1}$ exists as a bounded operator on $L^2(\R^3,\C^4)$ and maps $L^2(\R^3,\C^4)$ onto $H^1_A(\R^3,\C^4) = H^1_A(\R^3) \otimes \C^4$. The symbol of the resolvent 
			\begin{align*}
				\bigl ( H - z \bigr )^{(-1)^B} := {\Op^A}^{-1} \Bigl ( \bigl ( \Hd - z \bigr )^{-1} \Bigr ) 
			\end{align*}
			exists in $\mathcal{S}' \bigl ( T^* \R^3 , \mathcal{B}(\C^4) \bigr )$ by the Schwartz kernel theorem and is also an element of the magnetic Moyal algebra $\mathfrak{M}^B$ (cf.~\cite[Definition~2.1]{Iftimie_Mantoiu_Purice:commutator_criteria:2008}). 
			By \cite[Theorem~3.14]{Iftimie_Mantiou_Purice:magnetic_psido:2006}, $\mathfrak{s}_1 \otimes \id_{\C^4}$ defined as in \cite[equation~(5.3)]{Iftimie_Mantoiu_Purice:commutator_criteria:2008} is a continuous map from $H^1_A(\R^3,\C^4)$ to $L^2(\R^3,\C^4)$. Thus, the quantization of the product  
			\begin{align*}
				\bigl ( \mathfrak{s}_1 \otimes \id_{\C^4} \bigr ) \magW \bigl ( H - z \bigr )^{(-1)^B} 
			\end{align*}
			is an element of $\mathfrak{C}^B \otimes \mathcal{B}(\C^4)$ where $\mathfrak{C}^B := {\Op^A}^{-1} \Bigl ( \mathcal{B} \bigl ( L^2(\R^3) \bigr ) \Bigr )$. This means, the assumptions of \cite[Proposition~6.3]{Iftimie_Mantoiu_Purice:commutator_criteria:2008} are satisfied and we obtain $\bigl ( H - z \bigr )^{(-1)^B} \in \Hoermr{-1}{1} \bigl ( T^* \R^3 , \mathcal{B}(\C^4) \bigr )$. 
			\item Repeating the arguments in \cite[Proposition~6.7]{Iftimie_Mantoiu_Purice:commutator_criteria:2008} yields $\chi_B \in \Hoermr{0}{1}$. To see that $\chi_B$ is really a Hörmander symbol of order $-\infty$, one has to adapt the last part of the proof of \cite[Proposition~8.7]{Dimassi_Sjoestrand:spectral_asymptotics:1999}. 
		\end{enumerate}
	\end{proof}
	%
	% subsection Extension to matrix-valued symbols (end)

	\subsection{Asymptotic expansions of the magnetic Moyal product} % (fold)
	\label{appendix:mag_PsiDOs:asymptotic_expansions}
	As explained in the previous section, the results of \cite{Lein:two_parameter_asymptotics:2008} extend straightforwardly to matrix-valued symbols: there, the magnetic Weyl calculus associated to the building block operators 
	\begin{align}
		\mathsf{P}^A &= - i \eps \nabla_x - \lambda A(\hat{x}) 
		\label{appendix:mag_PsiDOs:eqn:building_block_ops_eps_lambda}
		\\
		\mathsf{Q} &= \hat{x} 
		\notag 
	\end{align}
	was considered. The associated Weyl product $\magW$ can then be expanded in $\eps$, $\lambda$ or $\eps$ \emph{and} $\lambda$ simultaneously. If $f$ and $g$ are two Hörmander symbols of order $m_1$ and $m_2$, and type $\rho \in (0,1]$, the expansion of $f \magW g$ in $\eps$ and $\lambda$ yields 
	\begin{align*}
		f \magW g \asymp \sum_{n = 0}^{\infty} \sum_{k = 0}^n \eps^n \lambda^k \, \bigl ( f \magW g \bigr )_{(n,k)}
	\end{align*}
	where the $\bigl ( f \magW g \bigr )_{(n,k)} \in \Hoermr{m_1 + m_2 - (n+k) \rho}{\rho}$ are known explicitly (cf.~\cite[Theorem~1.1]{Lein:two_parameter_asymptotics:2008}). 
	
	We will only need the first few terms of the expansion: the leading-order contribution is the point-wise product, 
	\begin{align*}
		(f \magW g)_{(0,0)} &= f \, g 
		,
	\end{align*}
	while the first-order correction in $\eps$ combine to the \emph{magnetic} Poisson bracket, 
	\begin{align*}
		(f \magW g)_{(1,0)} &= - \tfrac{i}{2} \bigl \{ f , g \bigr \} 
		:= - \frac{i}{2} \sum_{j = 1}^d \bigl ( \partial_{\xi_l} f \, \partial_{x_j} g - \partial_{x_j} f \, \partial_{\xi_j} g \bigr ) = - \tfrac{i}{2} \bigl \{ f , g \bigr \}  \\
		(f \magW g)_{(1,1)} &= + \frac{i}{2} \sum_{l , j = 1}^d B_{lj} \, \partial_{\xi_l} f \,  \partial_{\xi_j} g 
		. 
	\end{align*}
	Here, we have identified the magnetic field vector $\mathbf{B} = (\mathbf{B}_1,\mathbf{B}_2,\mathbf{B}_3)$ with the antisymmetric matrix 
	\begin{align*}
		B = \bigl ( B_{jk} \bigr )_{1 \leq j , k \leq 3} 
		= \left (
		\begin{matrix}
			0 & -\mathbf{B}_3 & \mathbf{B}_2 \\
			\mathbf{B}_3 & 0 & -\mathbf{B}_1 \\
			-\mathbf{B}_2 & \mathbf{B}_1 & 0 \\
		\end{matrix}
		\right )
		. 
	\end{align*}
	The second-order terms in $\eps$ contain at least three derivatives of $f$ and $g$ in momentum, 
	\begin{align*}
		(f \magW g)_{(2,0)} &= + \frac{1}{4} \sum_{l , j = 1}^d \Bigl ( 
		\partial_{\xi_l} \partial_{\xi_j} f \,  \partial_{x_l} \partial_{x_j} g + 
		\partial_{x_l} \partial_{x_j} f \,  \partial_{\xi_l} \partial_{\xi_j} g + 
		\Bigr . \\
		\Bigl . &\qquad \qquad \qquad - 
		\partial_{\xi_l} \partial_{x_j} f \,  \partial_{x_l} \partial_{\xi_j} g - 
		\partial_{x_l} \partial_{\xi_j} f \,  \partial_{\xi_l} \partial_{x_j} g 
		\Bigr ) \\
		(f \magW g)_{(2,1)} &= + \frac{1}{4} \sum_{j , k , l = 1}^d \Bigl ( \tfrac{1}{3} \partial_{x_j} B_{lk} \, \bigl ( 
		\partial_{\xi_l} \partial_{\xi_j} f \, \partial_{\xi_k} g -
		\partial_{\xi_l} f \,                  \partial_{\xi_j} \partial_{\xi_k} g \bigr ) + \Bigr . \\
		\Bigl . &\qquad \qquad \qquad - 
		B_{lk} \, \bigl ( 
		\partial_{\xi_l} \partial_{\xi_j} f \, \partial_{\xi_k} \partial_{x_j} g - 
		\partial_{\xi_l} \partial_{x_j} f \, \partial_{\xi_k} \partial_{\xi_j} g  \bigr ) 
		\Bigr ) \\
		(f \magW g)_{(2,2)} &= - \frac{1}{8} \sum_{j_1 , j_2 , l_1 , l_2 = 1}^d B_{l_1 j_1} \, B_{l_2 j_2} \, \partial_{\xi_{l_1}} \partial_{\xi_{l_2}} f \, \partial_{\xi_{j_1}} \partial_{\xi_{j_2}} g 
		. 
	\end{align*}
	For the computations, we will not need the explicit form of the second-order terms in $\eps$, just that for the $(2,1)$ and $(2,2)$ terms, either $f$ or $g$ are derived twice with respect to momentum.

	\subsubsection{Asymptotic expansion of $\magWsr$} % (fold)
	\label{appendix:mag_PsiDOs:expansion_magWsr}
	The two-parameter expansion in $\eps$ and $\lambda$ can be used to find an expansion of $\magWsr$ in $\nicefrac{1}{c}$: comparing the semi-relativistic building block operators \eqref{semirel:eqn:building_block_ops} to \eqref{appendix:mag_PsiDOs:eqn:building_block_ops_eps_lambda}, we see that $\eps$ is replaced by $\tfrac{\eps}{c}$ and $\lambda$ by $\tfrac{1}{c^2}$. This and a little book-keeping of Hörmander classes make up the proof of Lemma~\ref{semirel:lem:asymp_expansion_magW}: 
	\begin{proof}[Proof of Lemma~\ref{semirel:lem:asymp_expansion_magW}]
		Even though the proofs in \cite{Lein:two_parameter_asymptotics:2008} pertain to scalar-valued symbols, they can be adapted to the case where $f$ and $g$ are operator-valued symbols with obvious modifications. 
		\begin{enumerate}[(i)]
			\item The $n$th term of the expansion
			\begin{align*}
				\bigl ( f \magWsr g \bigr )_{(n)} = \sum_{3k \leq n} \eps^{n - 2k} \, \bigl ( f \magW g \bigr )_{(n-2k,k)}
			\end{align*}
			consists of a sum of terms where the $(n-2k,k)$th term is an element of $\Hoermr{m_1 + m_2 - (n - k) \rho}{\rho}$. Then since $k \leq \nicefrac{n}{3}$ and $\Hoermr{m'}{\rho} \subseteq \Hoermr{m}{\rho}$ for all $m' \leq m$ holds, we have shown $\bigl ( f \magWsr g \bigr )_{(n)} \in \Hoermr{m_1 + m_2 - (n - \nicefrac{n}{3}) \rho}{\rho} = \Hoermr{m_1 + m_2 - \frac{2}{3} n \rho}{\rho}$. 

			By \cite[Theorem~1.1]{Lein:two_parameter_asymptotics:2008}, the remainder of the two-parameter expansion is an element of $\Hoermr{m_1 + m_2 - (N+1) \rho}{\rho}$. Furthermore, the difference $\tilde{R}_N(f,g) - \eps^{N+1} R_N(f,g)$ consists of extraneous terms of the original asymptotic expansion such that $n \geq N+1$ and $3k \leq n$. Thus, we conclude that for all $n$ and $k$, the terms $\bigl ( f \magW g \bigr )_{(n - 2k,k)}$ that make up the difference are all contained in 
			\begin{align*}
				\Hoermr{m_1 + m_2 - (n - k) \rho}{\rho} \subseteq \Hoermr{m_1 + m_2 - (n - \nicefrac{n}{3}) \rho}{\rho} \subseteq \Hoermr{m_1 + m_2 - \frac{2}{3} (N+1) \rho}{\rho} 
			\end{align*}
			and the remainder is in the correct symbol class.
			\item This follows directly from (i) and remarking that for fixed $k + l + j = n$, the term $\bigl ( f_k \magW g_l \bigr )_{(j)}$ has $\tfrac{1}{c^{k + l + j}} = \tfrac{1}{c^n}$ as prefactor and is of symbol class $\Hoermr{m_1 + m_2 - \frac{2}{3}(k + l + j) \rho}{\rho} = \Hoermr{m_1 + m_2 - \frac{2}{3} n \rho}{\rho}$. 
			\item This is a straightforward consequence of the form of the remainder as given in (i) and \cite[Proof of Theorem~1.1]{Lein:two_parameter_asymptotics:2008}. 
		\end{enumerate}
	\end{proof}
	%
	% subsubsection Asymptotic expansion of $\magWsr$ (end)

	\subsubsection{Asymptotic expansion of $\magWnr$} % (fold)
	\label{appendix:mag_PsiDOs:expansion_magWnr}
	In the semi-relativistic scaling, $\eps$ in \eqref{appendix:mag_PsiDOs:eqn:building_block_ops_eps_lambda} remains $\eps$ and $\lambda$ becomes $\nicefrac{1}{c}$. The expansion of the product $\magWnr$ is then the expansion of $\magW$ in $\lambda$ (cf.~\cite[Theorem~2.12]{Lein:two_parameter_asymptotics:2008}). The following Corollary will be useful for the calculations in Appendix~\ref{appendix:non_rel}: 
	\begin{cor}\label{appendix:mag_PsiDOs:cor:asymp_expansion_magWnr}
		Let $f \in \Hoerr{m_1}$, $g \in \Hoerr{m_2}$, $m_1 , m_2 \in \R$, $\rho \in (0,1]$, and assume the components of $\mathbf{B}$ are $\Cont^{\infty}_{\mathrm{b}}$ functions. Then the following statements hold true: 
		\begin{enumerate}[(i)]
			\item If one of the factors, \eg $f \in \Cont^{\infty}_{\mathrm{b}}(\R^3)$, depends on position only, then $f \magWnr g = \bigl ( f \magWnr g \bigr )_{(0)} = f \WeylProd_{\mathrm{nr}}^{B = 0} g$ reduces to the non-magnetic Weyl product. 
			\item If $f , g \in \Cont^{\infty}_{\mathrm{b}}(\R^3)$ are functions of position only, then $f \magWnr g = f \, g$ reduces to the pointwise product of functions. 
			\item If $f$ or $g$ is a polynomial in $\xi$ of finite order, \eg if 
			\begin{align}
				f(x,\xi) = \sum_{\abs{\alpha} \leq m_1} b_{\alpha} \, \xi^{\alpha}
				, 
				\label{non_rel:eqn:polynomial_function}
			\end{align}
			then the expansion of $f \magWnr g$ terminates after at most $m_1$ terms and each $(f \magWnr g)_{(j)}$ can be written as a finite linear combination of the $(f \magW g)_{(n,j)}$ from \cite[Theorem~1.1]{Lein:two_parameter_asymptotics:2008} with $j \leq n \leq m_1$. 
		\end{enumerate}
	\end{cor}
	\begin{proof}
		\begin{enumerate}[(i)]
			\item An inspection of equation~(2.11) of \cite{Lein:two_parameter_asymptotics:2008} yields that except for $n = 0$, all other terms contain derivatives of $f$ and $g$ with respect to momentum. Hence, only the term for $n = 0$ survives. On the other hand, $f \magWnr g = \bigl ( f \magWnr g \bigr )_{(0)}$ coincides with the non-magnetic Weyl product $\WeylProd_{\mathrm{nr}}^{B = 0}$. 
			\item By (i), we have $f \magWnr g = f \WeylProd_{\mathrm{nr}}^{B = 0} g$. Then upon expanding the non-magnetic Moyal product (\eg by setting $\lambda = 0$, thus keeping only terms of the type $(f \magW g)_{(n,0)}$ in \cite[Theorem~1.1]{Lein:two_parameter_asymptotics:2008}), we see that $f \WeylProd_{\mathrm{nr}}^{B = 0} g = f \, g$. 
			\item The claim follows from Theorem 2.13 of \cite{Lein:two_parameter_asymptotics:2008} and the fact that all terms $(f \magW g)_{(n,j)}$ contain at least $j$ derivatives of $f$ and $g$ with respect to momentum. 
		\end{enumerate}
	\end{proof}
	%
	% subsubsection Asymptotic expansion of $\magWnr$ (end)
	% subsection Asymptotic expansions of the magnetic Moyal product (end)
	% section Magnetic $\Psi$DOs (end)

	\section{Calculations for the non-relativistic limit} % (fold)
	\label{appendix:non_rel}
	Calculating $\pi^{(3)}$, $u^{(3)}$ and $\heff$ up to fourth order is only possible because at least one of the functions in $(f \magWnr g)_{(n)}$ which occur in the various perturbation expansions is a polynomial in $\xi$ and independent of $x$. We have collected some of the necessary calculation rules in Corollary~\ref{appendix:mag_PsiDOs:cor:asymp_expansion_magWnr}. Furthermore, we will use the following notation to streamline presentation: the $n$th order term of the expansion of the Moyal commutator is abbreviated by 
	\begin{align*}
		\bigl [ f , g \bigr ]_{(n)} := \bigl ( f \magWnr g \bigr )_{(n)} - \bigl ( g \magWnr f \bigr )_{(n)} 
		. 
	\end{align*}
	Likewise, the pointwise commutator and anti-commutator of matrices are denoted by 
	\begin{align*}
		\bigl [ f , g \bigr ] &:= f \, g - g \, f 
		\\
		\bigl [ f , g \bigr ]_+ &:= f \, g + g \, f 
		. 
	\end{align*}

	\subsection{Almost Moyal projection $\pi^{(k)}$} % (fold)
	\label{appendix:non_rel:projection}
	Given $\pi^{(k)}$, we can compute the next-order term $\pi_{k+1}$ from the »projection defect« $G_{k+1}$ (see~\eqref{non_rel:eqn:projection_defect_comp}) and the »commutation defect« $F_{k+1}$ (see~\eqref{non_rel:eqn:commutation_defect_comp}) using equations~\eqref{non_rel:eqn:proj_diag_comp} and \eqref{non_rel:eqn:proj_offdiag_comp}, respectively. 
		
	Since $\pi_0$ is a constant matrix-valued function, the Moyal product $\magWnr$ reduces to the pointwise product of matrices and thus the projection defect of $\pi^{(0)} = \pi_0$ and consequently the diagonal part of $\pi_1$ both vanish identically. By direct computation, one obtains $F_1 = (\xi \cdot \alpha) \, \beta$ and thus $\pi_1$ is purely off-diagonal, 
	\begin{align*}
		\pi_1 = \pi_1^{\mathrm{od}} = \frac{1}{2m} \, \xi \cdot \alpha 
		. 
	\end{align*}
	Only one of the terms that make up the second-order projection defect survive, 
	\begin{align*}
		G_2 &= \bigl ( \pi_1 \magWnr \pi_1 \bigr )_{(0)} + \bigl ( \pi_0 \magWnr \pi_1 \bigr )_{(1)} + \bigl ( \pi_1 \magWnr \pi_0 \bigr )_{(1)} + \bigl ( \pi_0 \magWnr \pi_0 \bigr )_{(2)} 
		\\
		&= \pi_1^2 = \frac{1}{4 m^2} \, \xi^2 \, \id_{\C^4} 
		, 
	\end{align*}
	and seeing as $\pi_1$ depends only on momentum, $\bigl ( \pi_1 \magWnr \pi_1 \bigr )_{(0)} = \pi_1 \WeylProd^{B = 0}_{\mathrm{nr}} \pi_1 = \pi_1^2$ reduces to the pointwise product of matrices. Hence, the diagonal part of $\pi_2$ turns out to be 
	\begin{align*}
		\pi_2^{\mathrm{d}} &= - \frac{1}{4 m^2} \, \xi^2 \, \beta
		. 
	\end{align*}
	Exploiting the explicit form of the $H_j$ and $\pi_k$ as well as the properties of the asymptotic expansion of $\magWnr$ (Corollary~\ref{appendix:mag_PsiDOs:cor:asymp_expansion_magWnr}) yields 
	\begin{align*}
		F_2 &= \bigl [ H_1 , \pi_1 \bigr ]_{(0)} + \bigl [ H_2 , \pi_0 \bigr ]_{(0)} 
		+ \\
		&\qquad 
		+ \bigl [ H_0 , \pi_1 \bigr ]_{(1)} + \bigl [ H_1 , \pi_0 \bigr ]_{(1)} + \bigl [ H_0 , \pi_0 \bigr ]_{(2)} 
		% \\
		% &
		= 0 
	\end{align*}
	and thus the off-diagonal part of $\pi_2$ vanishes, 
	\begin{align*}
		\pi_2 = \pi_2^{\mathrm{d}} = - \frac{1}{4 m^2} \, \xi^2 \, \beta 
		. 
	\end{align*}
	In a similar vein, the third-order projection projection defect simplifies to 
	\begin{align*}
		G_3 &= \bigl [ \pi_1 , \pi_2 \bigr ]_+ + \bigl ( \pi_1 \magWnr \pi_1 \bigr )_{(1)}
		= \eps \bigl ( \pi_1 \magW \pi_1 \bigr )_{(1,1)} 
		\\
		&
		= \eps \frac{i}{8 m^2} \, \sum_{j , l = 1}^3 B_{jl} \; \partial_{\xi_j} (\xi \cdot \alpha) \; \partial_{\xi_l} (\xi \cdot \alpha)
		= - \frac{\eps}{4 m^2} \, \mathbf{B} \cdot \Sigma 
		, 
	\end{align*}
	and we obtain 
	\begin{align*}
		\pi_3^{\mathrm{d}} &= \frac{\eps}{4 m^2} \, (\mathbf{B} \cdot \Sigma) \beta
	\end{align*}
	for the diagonal part of $\pi_3$. Similarly, one can check that 
	\begin{align*}
		F_3 &= \bigl [ H_1 , \pi_2 \bigr ]_{(0)} + \bigl [ H_2 , \pi_1 \bigr ]_{(0)} 
		\\
		&= - \frac{1}{2m^2} \, \xi^2 \, (\xi \cdot \alpha) \beta  + \frac{\eps}{2m} \, i \nabla_x V \cdot \alpha 
	\end{align*}
	and therefore the off-diagonal part computes to be 
	\begin{align*}
		\pi_3^{\mathrm{od}} &= - \frac{1}{4 m^3} \, \xi^2 \, (\xi \cdot \alpha)  + \frac{\eps}{2 m^2} \, (i \nabla_x V \cdot \alpha) \beta 
		. 
	\end{align*}
	Combining diagonal and off-diagonal contributions yields 
	\begin{align*}
		\pi_3 &= \frac{\eps}{4 m^2} \, (\mathbf{B} \cdot \Sigma) \beta - \frac{1}{4 m^3} \, \xi^2 \, (\xi \cdot \alpha)  + \frac{\eps}{2 m^2} \, (i \nabla_x V \cdot \alpha) \beta 
		, 
	\end{align*}
	and we have shown equation~\eqref{non_rel:eqn:pi_3_explicit}. 
	% subsection »Projection« (end)

	\subsection{Almost Moyal unitary $u^{(k)}$} % (fold)
	\label{appendix:non_rel:unitary}
	From the unitarity defect $A_k$ (see~\eqref{non_rel:eqn:unitarity_defect_comp}) and intertwining defect $B_k$ (see~\eqref{non_rel:eqn:intertwining_defect_comp}), we can compute symmetric and anti-symmetric part of $u_k = a_k + b_k$ (corresponding to first and second term in \eqref{non_rel:eqn:u_k_1_comp}). 
	
	From $u_0 = \id_{\C^4}$, we easily compute 
	\begin{align*}
		u_1 &= 0 - \frac{1}{2m} \, (\xi \cdot \alpha) \beta 
		= - \frac{1}{2m} \, (\xi \cdot \alpha) \beta 
		. 
	\end{align*}
	The second-order term is purely symmetric, 
	\begin{align*}
		u_2 = - \tfrac{1}{2} \bigl ( u_1 \magWnr u_1^* \bigr )_{(0)} = - \tfrac{1}{2} {u_1}^2 = - \frac{1}{8 m^2} \, \xi^2 \, \id_{\C^4} 
		, 
	\end{align*}
	since the intertwining defect happens to vanish identically, 
	\begin{align*}
		B_2 &= \bigl ( \pi_0 \magWnr a_2 \bigr )_{(0)} + \bigl ( a_2 \magWnr \pi_0 \bigr )_{(0)} + \bigl ( u_1 \magWnr \pi_1 \bigr )_{(0)} + \bigl ( \pi_1 \magWnr u_1 \bigr )_{(0)} 
		+ \\
		&\qquad 
		+ \Bigl ( \bigl ( u_0 \magWnr \pi_2 \bigr )_{(0)} \magWnr u_0^* \Bigr )_{(0)} + \Bigl ( \bigl ( u_1 \magWnr \pi_0 \bigr )_{(0)} \magWnr u_1^* \Bigr )_{(0)} 
		\\
		&= \bigl [ \pi_0 , a_2 \bigr ]_+ + \bigl [ u_1 , \pi_1 \bigr ] + \pi_2 - u_1 \, \pi_0 \, u_1 
		\\
		&= \frac{1}{4 m^2} \xi^2 \Bigl ( - \piref + 2 \beta - \beta + \bigl ( \id_{\C^4} - \piref \bigr ) \Bigr )
		= 0 
		. 
	\end{align*}
	Obtaining the third-order projection defect is still relatively easy: of 8 terms, only 3 are non-zero, and we get 
	\begin{align*}
		G_3 &= \bigl ( u_1 \magWnr u_2^* \bigr )_{(0)} + \bigl ( u_2 \magWnr u_1^* \bigr )_{(0)} + \bigl ( u_1 \magWnr u_1^* \bigr )_{(1)} 
		\\
		&= \bigl [ u_1 , u_2 \bigr ] - \bigl ( u_1 \magW u_1 \bigr )_{(1,1)}
		= 0 - \frac{\eps}{8 m^2} \, \sum_{j , k = 1}^3 B_{jk} \, \partial_{\xi_j} \bigl ( (\xi \cdot \alpha) \beta \bigr ) \, \partial_{\xi_k} \bigl ( (\xi \cdot \alpha) \beta \bigr ) 
		\\
		&= - \frac{\eps}{4 m^2} \, \mathbf{B} \cdot \Sigma 
		. 
	\end{align*}
	This implies the symmetric contribution to $u_3$ is $a_3 = + \frac{\eps}{8 m^2} \, \mathbf{B} \cdot \Sigma$. Computing the intertwining defect is a bit more involved: of 44 terms, only twelve are potentially not equal to $0$. Pairing some of these remaining terms in commutators and anticommutators allows us to calculate $B_3$ more easily: 
	\begin{align*}
		B_3 &= \bigl ( \pi_0 \magWnr a_3 \bigr )_{(0)} + \bigl ( a_3 \magWnr \pi_0 \bigr )_{(0)} 
		+ \bigl ( u_1 \magWnr \pi_2 \bigr )_{(0)} + \bigl ( \pi_2 \magWnr u_1^* \bigr )_{(0)} 
		+ \\
		&\qquad 
		+ \bigl ( \pi_1 \magWnr u_2^* \bigr )_{(0)} + \bigl ( u_2 \magWnr \pi_1 \bigr )_{(0)} 
		+ \\
		&\qquad 
		+ \Bigl ( \bigl ( u_1 \magWnr \pi_0 \bigr )_{(0)} \magWnr u_2^* \Bigr )_{(0)} + \Bigl ( \bigl ( u_2 \magWnr \pi_0 \bigr ) \magWnr u_1^* \Bigr )_{(0)}
		+ \\
		&\qquad 
		+ \Bigl ( \bigl ( u_0 \magWnr \pi_3 \bigr )_{(0)} \magWnr u_0^* \Bigr )_{(0)}
		+ \Bigl ( \bigl ( u_1 \magWnr \pi_1 \bigr )_{(0)} \magWnr u_1^* \Bigr )_{(0)}
		+ \\
		&\qquad 
		+ \bigl ( u_1 \magWnr \pi_1 \bigr )_{(1)} + \bigl ( \pi_1 \magWnr u_1^* \bigr )_{(1)} 
		\\
		&= \bigl [ \pi_0 , a_3 \bigr ]_+ + \bigl [ u_1 , \pi_2 \bigr ] + \bigl [ u_2 , \pi_1 \bigr ]_+ 
		+ \\
		&\qquad 
		+ \bigl [ u_1 , \pi_0 \bigr ] \, u_2 + \pi_3 - u_1 \, \pi_1 \, u_1 + \bigl [ u_1 , \pi_1 \bigr ]_{(1)}
		\\
		&= - \frac{3}{16 m^3} \, \xi^2 \, (\xi \cdot \alpha) + \frac{\eps}{4 m^2} \, (i \nabla_x V \cdot \alpha) \beta - \frac{\eps}{4 m^2} \, (\mathbf{B} \cdot \Sigma) ( \pi_0 + \beta )
	\end{align*}
	Since $\beta$ and $\pi_0 = \piref$ commute with $\piref$, the last term does not contribute, and we obtain 
	\begin{align*}
		b_3 &= \bigl [ B_3 , \piref \bigr ] 
		\\
		&= + \frac{3}{16 m^3} \, \xi^2 \, (\xi \cdot \alpha) \beta + \frac{\eps}{4 m^2} \, (i \nabla_x V \cdot \alpha) 
	\end{align*}
	for the antisymmetric part of $u_3$. Put together, we get 
	\begin{align*}
		u_3 = \frac{\eps}{8 m^2} \, \mathbf{B} \cdot \Sigma + \frac{3}{16 m^3} \, \xi^2 \, (\xi \cdot \alpha) \beta + \frac{\eps}{4 m^2} \, (i \nabla_x V \cdot \alpha) 
	\end{align*}
	for the third-order correction. This proves equation~\eqref{non_rel:eqn:u_3_explicit}. 
	% subsection »Projection« (end)

	\subsection{Effective hamiltonian $\heff$} % (fold)
	\label{appendix:non_rel:heff}
	The diagonalized hamiltonian $h := \bigl ( u \magWnr \Hnr \magWnr u^* \bigr )^{(4)}$ can be computed recursively from 
	\begin{align}
		h_{k} &= \bigl ( h_{k} \magWnr u_0 \bigr )_{(0)} 
		% \notag \\
		% &
		= c^k \bigl ( u \magWnr \Hnr - h^{(k-1)} \magWnr u \bigr ) + \order(\nicefrac{1}{c}) 
		\notag \\
		&
		= H_k + \sum_{\substack{j + l + n = k \\ l \leq k-1}} \Bigl ( \bigl ( u_j \magWnr H_l \bigr )_{(n)} - \bigl ( h_l \magWnr u_j \bigr )_{(n)} \Bigr )
		+ \order(\nicefrac{1}{c})
		\label{appendix:non_rel:recursion_relation_h_k}
	\end{align}
	which has the benefit that one only needs to compute a \emph{single} Moyal product rather than a \emph{double} Moyal product. Due to the specific structure of the terms and the non-relativistic pseudodifferential calculus, many of the terms vanish identically. 
	
	Since $u_0 = \id_{\C^4}$, the leading-order terms equals $h_0 = H_0 = m \beta$. The recursion for $h_1$ yields 
	\begin{align*}
		h_1 &= H_1 + \bigl ( u_1 \magWnr H_0 \bigr )_{(0)} - \bigl ( h_0 \magWnr u_1 \bigr )_{(0)} + \bigl ( u_0 \magWnr H_0 \bigr )_{(1)} - \bigl ( h_0 \magWnr u_0 \bigr )_{(1)} 
		\\
		&
		= \xi \cdot \alpha - \frac{1}{2m} (\xi \cdot \alpha) \beta \, m \beta - m \beta \, \left ( - \frac{1}{2 m} \right ) (\xi \cdot \alpha) \beta 
		= 0 
		. 
	\end{align*}
	The second-order term is also easily obtained, 
	\begin{align*}
		h_2 &= H_2 + \bigl ( u_1 \magWnr H_1 \bigr )_{(0)} + \bigl [ u_2 , H_0 \bigr ]_{(0)} 
		\\
		&= V - \frac{1}{2 m} (\xi \cdot \alpha) \beta (\xi \cdot \alpha)
		= \frac{1}{2m} \, \xi^2 \, \beta + V 
		. 
	\end{align*}
	Of 18 terms which comprise the third-order term, all but 6 vanish since they contain $u_0$, $h_0$, $h_1 = 0$ or $H_0$ which are constant functions of $x$ and $\xi$. Computing the remaining terms after pairing them whenever appropriate yields 
	\begin{align*}
		h_3 &= \bigl ( u_1 \magWnr H_2 \bigr )_{(0)} - \bigl ( h_2 \magWnr u_1 \bigr )_{(0)} + \bigl ( u_2 \magWnr H_1 \bigr )_{(0)} + \bigl [ u_3 , H_0 \bigr ]_{(0)} + \bigl ( u_1 \magWnr H_1 \bigr )_{(1)} 
		\\
		&= \bigl [ u_1 , V \bigr ]_{(0)} - \frac{1}{2m} \, \xi^2 \, \beta \, u_1 + u_2 \, H_1 + \bigl [ u_3 , H_0 \bigr ] + \bigl ( u_1 \magWnr H_1 \bigr )_{(1)} 
		\\
		&= \frac{\eps}{2m} \, (i \nabla_x V \cdot \alpha) \beta 
		- \frac{1}{4 m^2} \, \xi^2 \, (\xi \cdot \alpha)
		- \frac{1}{8 m^2} \, \xi^2 \, (\xi \cdot \alpha) 
		+ \\
		&\qquad 
		+ \frac{3}{8 m^2} \, \xi^2 \, (\xi \cdot \alpha) - \frac{\eps}{2m} \, (i \nabla_x V \cdot \alpha) \beta 
		- \frac{\eps}{2 m} \, (\mathbf{B} \cdot \Sigma) \beta 
		\\
		&= - \frac{\eps}{2 m} \, (\mathbf{B} \cdot \Sigma) \beta 
		. 
	\end{align*}
	We do not need to calculate $h_4$, but only $h_{\mathrm{eff} \, 4}$. This simplifies the problem tremendously, because (i) we may leave out terms which are off-diagonal with respect to the splitting induced by $\piref$ and (ii) we do not need to \emph{compute} $\pi_4$ and $u_4$. Note that the \emph{existence} of $\pi_4$ and $u_4$ are guaranteed by Propositions~\ref{non_rel:prop:existence_pi} and \ref{non_rel:prop:existence_u}, respectively. Again, using that $H_0 = h_0$ and $u_0$ are constant symbols as well as $h_1 = 0$ cuts down the number of terms for $h_{\mathrm{eff} \, 4}$ from 26 to 9, 
	\begin{align*}
		h_{\mathrm{eff} \, 4} &= \piref \, h_4 \, \piref 
		\\
		&= \piref \, \Bigl ( 
		\bigl ( u_4 \magWnr H_0 \bigr )_{(0)} - \bigl ( h_0 \magWnr u_4 \bigr )_{(0)} 
		+ \bigl ( u_3 \magWnr H_1 \bigr )_{(0)} 
		+ \Bigr . 
		\\
		&\qquad \qquad \Bigl . 
		+ \bigl ( u_2 \magWnr H_2 \bigr )_{(0)} - \bigl ( h_2 \magWnr u_2 \bigr )_{(0)} 
		- \bigl ( h_3 \magWnr u_1 \bigr )_{(0)} 
		+ \Bigr . 
		\\
		&\qquad \qquad \Bigl . 
		+ \bigl ( u_2 \magWnr H_1 \bigr )_{(1)} 
		+ \bigl ( u_1 \magWnr H_2 \bigr )_{(1)} 
		+ \bigl ( u_1 \magWnr H_1 \bigr )_{(2)} 
		\Bigr ) \, \piref 
		\\
		&= \piref \, \Bigl ( \bigl [ u_4 , H_0 \bigr ] + \bigl ( u_3 \magWnr H_1 \bigr )_{(0)} + \bigl [ u_2 , V \bigr ]_{(0)} - \tfrac{1}{2 m} \xi^2 \, \beta \, u_2 
		+ \Bigr . 
		\\
		&\qquad \qquad \Bigl . 
		- \bigl ( h_3 \magWnr u_1 \bigr )_{(0)} 
		+ \bigl ( u_2 \magWnr H_1 \bigr )_{(1)} 
		% + (u_1 \magWnr H_2)_{(1)} + (u_1 \magWnr H_1)_{(2)} 
		\Bigr ) \, \piref 
		. 
	\end{align*}
	The reason we do not need to compute $u_4$ is related to the fact that $h_{\mathrm{eff} \, 0} = m \, \piref = H_0 \piref$ is scalar-valued and thus the commutator 
	\begin{align*}
		\piref \, \bigl [ u_4 , H_0 \bigr ] \, \piref = \piref \, \bigl [ u_4 , h_{\mathrm{eff} \, 0} \bigr ] \, \piref 
		= 0
	\end{align*}
	vanishes identically. The next two terms 
	\begin{align*}
		\piref \, \bigl ( h_3 \magWnr u_1 \bigr )_{(0)} \, \piref &= 0 
		\\
		\piref \, \bigl ( u_2 \magWnr H_1 \bigr )_{(1)} \, \piref &= 0 
	\end{align*}
	also do not contribute since these products yield completely off-diagonal-matrix-valued functions. Lastly, the term 
	\begin{align*}
		\bigl ( u_1 \magWnr H_1 \bigr )_{(2)} = 0 
	\end{align*}
	is identically equal to $0$ even before projecting onto the upper-left block matrix. This is because $u_1$ and $H_1$ are linear in $\xi$ and independent of $x$ while second-order term of the expansion of $\magWnr$ contains only second-order derivatives of $u_1$ and $H_1$. 
	
	Now to the computation of the non-trivial terms: keeping only block diago\-nal terms, we obtain 
	\begin{align*}
		\piref \, \bigl ( u_3 &\magWnr H_1 \bigr )_{(0)} \, \piref 
		= \\
		&= \biggl ( - \frac{3}{16} \, \abs{\xi}^4 - \frac{\eps}{4 m^2} (i \nabla_x V \cdot \xi) + \frac{\eps}{4 m^2} (\nabla_x V \wedge \xi) \cdot \Sigma + \frac{\eps^2}{8 m^2} \Delta V 
		\biggr ) \, \piref 
		. 
	\end{align*}
	Using that $u_2$ is a second-order polynomial in $\xi$ and independent of $x$ as well as that $V$ depends only on position and is proportional to $\id_{\C^4}$, we conclude only one term of the $\eps$ expansion of 
	\begin{align*}
		\piref \, \bigl [ u_2 , V \bigr ]_{(0)} \, \piref &= - \eps \, i \, \piref \, \bigl \{ u_2 , V \bigr \} \, \piref 
		\\
		&= \frac{\eps}{4 m^2} \, (i \nabla_x V \cdot \xi) \, \piref 
	\end{align*}
	contributes. Last, but not least, keeping only block diagonal terms, we get 
	\begin{align*}
		\piref \, \beta \, u_2 \, \piref &= \piref \, u_2 \, \piref 
		= - \frac{1}{8 m^2} \, \xi^2 \, \piref 
		. 
	\end{align*}
	Combining all these terms, we obtain the fourth-order term of $\heff$, 
	\begin{align*}
		h_{\mathrm{eff} \, 4} &= \biggl ( \frac{-3+1}{16 m^3} \abs{\xi}^4 
		+ \frac{\eps}{4 m^2} \, (\nabla_x V \wedge \xi) \cdot \Sigma 
		\, + \biggr . \\
		&\qquad \qquad \biggl . 
		+ \frac{\eps \, (1 - 1)}{4 m^2} \, (i \nabla_x V \cdot \xi) + \frac{\eps^2}{8 m^2} \, \Delta V
		\biggr ) \, \piref 
		\\
		&= \biggl ( - \frac{1}{8 m^3} \abs{\xi}^4 + \frac{\eps}{4 m^2} \, (\nabla_x V \wedge \xi) \cdot \Sigma + \frac{\eps^2}{8 m^2} \, \Delta V
		\biggr ) \, \piref 
		, 
	\end{align*}
	thereby showing equation~\eqref{non_rel:eqn:heff_explicit}. 
	% subsection Effective hamiltonian (end)
	% section Non-relativistic case (end)
\end{appendix}

\printbibliography

\end{document}